\documentclass[pdflatex,sn-mathphys-num, iicol]{sn-jnl}

\usepackage[english]{babel}
\usepackage[utf8]{inputenc}

\usepackage{algorithm}
\usepackage{algorithmic}
\usepackage{amsfonts, amssymb, amsmath, amsthm}
\usepackage{cases}
\usepackage{epsfig}
\usepackage[T1]{fontenc}
\usepackage{hyperref}
\usepackage{moresize}
\usepackage[percent]{overpic}
\usepackage[usenames,dvipsnames]{pstricks}
\usepackage{pgf,tikz}
\usepackage{tkz-graph}

\usepackage{graphicx}
\usepackage{subcaption}

\theoremstyle{plain}
\newtheorem{theorem}{Theorem}[section]

\newtheorem*{conjecture*}{Conjecture}

\newtheorem{corollary}[theorem]{Corollary}
\newtheorem{definition}[theorem]{Definition}
\newtheorem{remark}[theorem]{Remark}

\renewcommand\theenumi{(\roman{enumi})}

\RequirePackage{amsmath}
\RequirePackage{xspace}
\RequirePackage{bbm}

\def\XS{\xspace}
\DeclareMathAlphabet{\mathb}{OML}{cmm}{b}{it}
\def\sbm#1{\ensuremath{\mathb{#1}}}                
\def\sbmm#1{\ensuremath{\boldsymbol{#1}}}          
\def\sbv#1{\ensuremath{\mathbf{#1}}}               
\def\scu#1{\ensuremath{\mathcal{#1\XS}}}           
\def\sbl#1{\ensuremath{\mathbb{#1}}}
\def\sbf#1{\ensuremath{\mathsf{#1}}}

\def\Lb{{\sbm{L}}\XS}  
\def\Mb{{\sbm{M}}\XS}

  \def\ub{{\sbm{u}}\XS}
  \def\vb{{\sbm{v}}\XS}
  \def\wb{{\sbm{w}}\XS}
  \def\xb{{\sbm{x}}\XS}
  \def\yb{{\sbm{y}}\XS}

\def\Bc{{\scu{B}}\XS}   
\def\Cc{{\scu{C}}\XS}

\def\Gc{{\scu{G}}\XS}

\def\Sc{{\scu{S}}\XS}

\def\Cbb{{\sbl{C}}\XS}

\def\Rbb{{\sbl{R}}\XS}

        \def\Phib      {{\sbmm{\Phi}}\XS}

        \def\Psib      {{\sbmm{\Psi}}\XS}
\def\omegab      {{\sbmm{\omega}}\XS}

\def\zerob   {{\sbv{0}}\XS}

\def\eC{\Cbb}

\def\eR{\Rbb}

  \def\ds{{\sbf{d}}\XS}

\def\Ps{{\sbf{P}}\XS}

\newcommand{\RR}{\ensuremath{\mathbb{R}}}
\newcommand{\NN}{\ensuremath{\mathbb N}}

\newcommand{\II}{\ensuremath{\mathbb I}}

\newcommand{\minimize}[2]{\ensuremath{\underset{\substack{{#1}}}
{\mathrm{minimize}}\;\;#2 }}

\newcommand{\Argmind}[2]{\ensuremath{\underset{\substack{{#1}}}%
{\mathrm{Argmin}}\;\;#2 }}

\newcommand{\inter}{\ensuremath{\raisebox{-0.5mm}{\mbox{{$\cap$}}}}}

\renewcommand{\leq}{\ensuremath{\leqslant}}
\renewcommand{\geq}{\ensuremath{\geqslant}}
\renewcommand{\le}{\ensuremath{\leqslant}}
\renewcommand{\ge}{\ensuremath{\geqslant}}

\newcommand{\emp}{\ensuremath{{\varnothing}}}

\newcommand{\Card}[1]{\ensuremath{\sharp \, #1}}
\newcommand{\menge}[2]{\big\{{#1}~\big |~{#2}\big\}}

\usetikzlibrary{arrows, positioning, decorations.markings}
\usetikzlibrary{shapes,decorations}
\tikzstyle{every picture}+=[remember picture]
\tikzstyle{mybox} = [draw=lightgreen, fill=white, very thick,
    rectangle, rounded corners, inner sep=5pt, inner ysep=5pt]
\tikzstyle{fancytitle} =[fill=lightgreen!30, text=black, ellipse]
\tikzstyle{myquestionbox} = [draw=orange, fill=orange!2, very thick, rectangle, rounded corners, inner sep=5pt, inner ysep=7pt]
\tikzstyle{fancyquestiontitle} =[fill=orange!80, text=white, ellipse]
\tikzstyle{mydefbox} = [rectangle,draw=black, rounded corners, very thick, inner sep=5pt, inner ysep=3pt]
\tikzstyle{myalgobox} = [rectangle,draw=lightgreen, fill=lightgreen!5, rounded corners, very thick, inner sep=5pt, inner ysep=3pt]
\tikzstyle{mythmbox} = [rectangle,draw=orange, fill=orange!5, rounded corners, very thick, inner sep=5pt, inner ysep=3pt]
\tikzstyle{notmythmbox} = [rectangle,draw=orange, fill=white, rounded corners, very thick, inner sep=5pt, inner ysep=3pt]
\tikzstyle{fancythmtitle} =[fill=orange!80, text=white, ellipse]
\tikzstyle{important_col} = [rectangle,draw=orange, fill=orange!5, rounded corners, thick, inner sep=1pt, inner ysep=2pt,text centered]

\newcommand{\tca}{{\widetilde{\mathcal{C}}_\alpha}}

\begin{document}

\title[Article Title]{
A Plug-and-Play Method with Inpainting Network for Bayesian Uncertainty Quantification in Imaging}



\author[1,3]{\fnm{Xiaoyu} \sur{Wang}}\email{xiaoyu.wang@hw.ac.uk}
\equalcont{These authors contributed equally to this work.}

\author[2,3]{\fnm{Michael} \sur{Tang}}\email{michael.cy.tang@gmail.com}
\equalcont{These authors contributed equally to this work.}

\author*[1,3]{\fnm{Audrey} \sur{Repetti}}\email{a.repetti@hw.ac.uk}

\affil[1]{\orgdiv{School of Mathematics and Computer Sciences and School of Engineering and Physical Sciences}, \orgname{Heriot-Watt University}, \orgaddress{\city{Edinburgh}, \country{UK}}}

\affil[2]{\orgdiv{School of Mathematics}, \orgname{University of Edinburgh}, \orgaddress{\city{Edinburgh}, \country{UK}}}

\affil[3]{\orgname{Maxwell Institute for Mathematical Sciences}, \orgaddress{\city{Edinburgh}, \country{UK}}}

\abstract{
We contribute to an uncertainty quantification problem in imaging that evaluates a hypothesis test questioning the existence of local ``artefacts'' appearing in the maximum \textit{a posteriori} (MAP) estimate (obtained from standard numerical tools). Such a method, called Bayesian uncertainty quantification by optimization (BUQO), was introduced a few years ago as an efficient and scalable alternative to sampling methods when per-pixel error-bars are not needed. 
BUQO formulates a hypothesis test for probing the existence of local structures in the MAP estimate as a minimization problem, that can be solved efficiently with standard optimization algorithms. 
In this context, BUQO requires a ``mathematical'' definition of the ``local artefact''. 
This definition can be interpreted as an inpainting of the structure. However, only simple hand-crafted techniques have been proposed so far due to the complexity of the problem.
In this work, we propose a data-driven alternative to BUQO where the inpainting procedure in the algorithm is performed using a convolutional inpainting neural network (NN).
This results in a plug-and-play algorithm, based on the primal-dual Condat-V\~{u} iterations, where the inpainting procedure is performed with a NN. 
The proposed approach is assessed on two image reconstruction problems inspired by medicine. We specifically perform simulations on two Fourier undersampling problems (discrete and non-uniform) encountered in magnetic resonance imaging, as well as a computed tomography problem using the Radon measurement operator.

}




\maketitle

\section{Introduction}
Imaging problems across many modalities can be stated as inverse problems. The aim is to find an estimate $\xb^\dagger \in \eR^N$ of an original unknown image $\overline{\xb} \in \eR^N$, from degraded measurements $\yb \in \eC^M$, observed through a forward model of the form
\begin{equation}\label{eq:model}
\yb = \Phib \overline{\xb} + \wb.
\end{equation}
Here $\Phib \colon \eR^N \to \eC^{M}$ is the measurement operator, which we assume to be known and linear, and $\wb\in \eC^M$ is a realization of an independent identically distributed (iid) random variable. 
In general, this inverse problem is ill-posed and/or ill-conditioned, so the estimate $\xb^\dagger$ cannot be obtained with a direct inversion model, and iterative approaches are required \cite{chambolle2016introduction, benning2018modern}. 

Following a Bayesian framework, the posterior distribution combines likelihood $p (\yb | \xb)$ and prior $p (\xb)$ distributions using the Bayes' theorem, i.e.,
\begin{equation*}
p (\xb | \yb) \propto p (\yb | \xb) p (\xb).
\end{equation*}
Furthermore, similarly to \cite{RPW18,Per17}, we assume log-concave uniform likelihood and \textit{prior} distribution, i.e.,  $p (\yb | \xb) = \exp( - f_\yb(\Phib \xb))$ and $p (\xb) = \exp(-g(\xb))$ where $f_\yb$ and $g$ are convex. Then
the posterior distribution of the problem can be expressed as 
\begin{equation}\label{eq:post-def}
p (\xb | \yb) \propto \exp\Big( -f_\yb(\Phib \xb) - g(\xb) \Big) ,
\end{equation}
where
the likelihood 
$f_\yb(\Phib \, \cdot) $ is associated with the statistics of the forward model~\eqref{eq:model}, and the \textit{prior} $ g$ is used to incorporate \textit{a priori} information known about the image of interest, to help to overcome ill-posedness and/or ill-conditionedness of the inverse problem. 
Classical choices include feasibility constraints (e.g., positivity for intensity images), or functions promoting either smoothness or sparsity (e.g., $\ell_2$ or $\ell_1$ norms) possibly in some transformed domain such as wavelet, Fourier or total variation (TV) (see e.g., \cite{Mallat, ROF1992}).
A classical inference approach is to find a maximum \textit{a posteriori} (MAP) estimator, defined as
\begin{equation}\label{eq:MAP}
\xb^\dagger = \Argmind{\xb\in \eR^N} f_{\yb}( \Phib\xb) + g(\xb).
\end{equation}
Problem~\eqref{eq:MAP} can be solved efficiently using optimization algorithms. Since early 2000's, proximal algorithms have been the state-of-the-art in imaging problems \cite{combettes2011proximal}. They are scalable to high-dimensional problems, and versatile to handle sophisticated problems involving non-smooth functions and linear operators. During the last decade, these algorithms have been significantly improved by the adoption of deep learning. In particular, in proximal algorithms, the operator handling the \textit{prior} term can be replaced by a denoising neural network (NN), leading to ``plug-and-play'' (PnP) methods~\cite{Zhang17, Zhang19, Ryu2019, Hurault2021, Schonlieb2018, Ahmad2020, Hertrich2021, Terris21, PRTW21, RTWP2022}. 

Although proximal optimization algorithms are very efficient and widely used for obtaining MAP estimates, they provide a point estimate only, without any additional information. In the absence of the ground truth, these approaches do not quantify any of the inherent uncertainties in such estimates. Uncertainty quantification (UQ) provides rigorous additional information to accompany qualitative estimates, and the rich theory of Bayesian inference can provide such analyses \cite{Robert2004, Robert2007}. UQ can be of great benefit with respect to decision making, for example using medical magnetic resonance (MR) or computed tomography (CT) scans to plan a course of treatment, where it is very important to know whether or not specific structures in a MAP estimate can be trusted to a quantified measure of certainty. 

Classically, sampling methods such as Markov chain Monte Carlo (MCMC) are used to generate estimators of the original unknown image by drawing random samples from the posterior distribution, and hence approximating the true distribution of interest \cite{Robert2007, Robert2004}. 
Such methods allow to perform UQ by evaluating confidence intervals and hypothesis testing. 
Unfortunately, traditional MCMC methods are computationally expensive, preventing their use in many high dimensional imaging applications. Over the past few years, many authors have proposed to use optimization ``tricks'' on MCMC methods to make them more scalable and better adapted to imaging. For instance, enabling the use of non-smooth functions \cite{Pereyra16, Pereyra20, Pereyra21}, using splitting techniques \cite{Vono21, Thouvenin22}, or using learned \textit{priors} \cite{laumont2022bayesian, holden2022bayesian, coeurdoux2024plug}. Although these methods do offer more scalability to draw each sample, they still necessitate a substantial number of samples to achieve an accurate approximation of the distribution of interest (due to the size $N$ of the image, i.e., the unknown parameters) \cite{biquard2025variational,zach2025statistical,thong2024bayesian}.

Deep learning based UQ methods have recently emerged as powerful tools in computational imaging, enabling systematic characterization of uncertainty inherent in ill-posed imaging problems \cite{siahkoohi2020deep, ekmekci2022uncertainty, tanno2019uncertainty, abdar2021review, zhu2023denoising, kawar2022denoising, biquard2025variational, crafts2025benchmarking, munier2025jackpot, zach2025statistical}. 
Recent evaluations however suggest that despite their remarkable image estimation accuracy, Bayesian deep learning approaches may fail to deliver reliable UQ results~\cite{thong2024bayesian}. 
Further, recent evaluations suggest that high quality image reconstruction does not necessarily translate into equally reliable UQ results, and that the quality of the resulting posterior approximation can depend significantly on adequate model choices~\cite{thong2024bayesian}.


Existing approaches often rely on yielding uncertainty measures by generating pixel-wise uncertainty or variance maps through Bayesian posterior sampling. 
Recently a different framework has been proposed,
to perform hypothesis test directly on the MAP estimate. This method, referred to as BUQO \cite{RPW18, Per17, RPW19, PereyraMNRAS1, PereyraMNRAS2, PNAS2022, liaudat2024scalable} enables scalable Bayesian UQ for imaging relying on optimisation-based methods. Unlike MCMC methods, this method does not provide per-pixel error-bars, but performs hypothesis tests to question the existence of local structures appearing in the MAP estimate (see an example in Figure~\ref{fig:pnp-buqo-demo}). Nevertheless, BUQO can also be modified to obtain ``per-patch'' error-bars \cite{PereyraMNRAS2, liaudat2024scalable}. 
In a nutshell, BUQO ``works'' on the pre-computed MAP estimate, formulating the hypothesis test to probe a local structure as a minimization problem. This can then be solved by leveraging the same scalable proximal algorithms as used to obtain the MAP estimate. 
Thus, such an approach is more scalable than traditional MCMC algorithms, particularly in high dimensional settings where computing the MAP is much cheaper than estimating the posterior distribution. 
Further, the practical BUQO framework is different from standard UQ methods (including MCMC and deep learning based methods), as it directly works on the MAP estimate. This particularity can be useful in practice as BUQO can be used on local computers by practicians without needing to be implemented on the imaging devices themselves (e.g., on scanners or telescopes).
Nevertheless, a precise mathematical definition identifying what is a ``structure'' in the MAP estimate is crucial, and highly dependent on the type of structure of interest. In practice, this definition can be interpreted as an inpainting task, but only simple hand-crafted techniques have been proposed in the literature so far, due to the complexity of the problem.

In this work, we propose a data-driven approach that employs a convolutional neural network (CNN) for the inpainting step within the BUQO framework. This allows for a more versatile formulation of the original minimization problem proposed in \cite{RPW19}, to perform the hypothesis test. 
In this formulation, the inpainting CNN explicitly appears in a squared $\ell_2$ function, that will be handled exclusively through its gradient in a PnP algorithm based on a primal-dual forward-backward algorithm \cite{Condat13, Vu13, Komodakis2014}.
While our method employs the network for inpainting rather than denoising, as done in standard PnP schemes, the idea of ``plugging in'' the inpainting CNN remains very similar. Indeed, the objective is to replace a task that could be done by standard optimization tools, such as a gradient or proximity operator, by a NN specifically trained for the task. We highlight the performances of the proposed PnP-BUQO approach, to perform hypothesis testing on both simulated MR and CT imaging problems. In particular, we show that the proposed PnP-BUQO enables more realistic visual results compared to the original BUQO inpainting method. 

The remainder of the article is organized as follows. Section~\ref{section:Bayes} recalls the necessary background in Bayesian inference and the classical BUQO hypothesis testing framework. In Section~\ref{section:problem} we introduce the fixed-point inpainting operator strategy and formulate our proposed plug-and-play BUQO approach. Section~\ref{section:simulations} details the numerical simulation setup. In Section~\ref{section:simulation-results}, we present a comprehensive collection of numerical results and comparative analysis. Finally, we conclude in Section~\ref{section:discussion}, discuss potential extensions and outline avenues for future work.


\begin{figure}[t]
    \centering
    \includegraphics[width=0.5\linewidth]{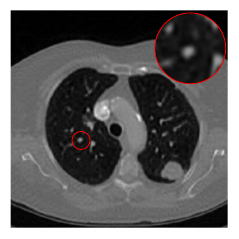}
    \hspace{-3mm}
    \includegraphics[width=0.5\linewidth]{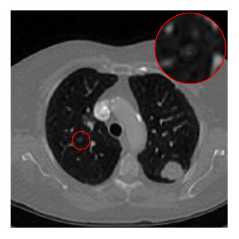}
    \caption{\small
    Illustrative example showing (Left) a MAP estimate $\xb^\dagger$ for CT reconstruction, with zoomed in region showing a local structure, and (Right) the same MAP estimate without the structure, denoted by $\xb^\dagger_{\Sc}$. 
    Objective is to determine whether the structure is true, or a reconstruction artefact.}
    \label{fig:pnp-buqo-demo}
\end{figure}

\begin{figure*}[t]
\centering
\includegraphics[scale=0.5]{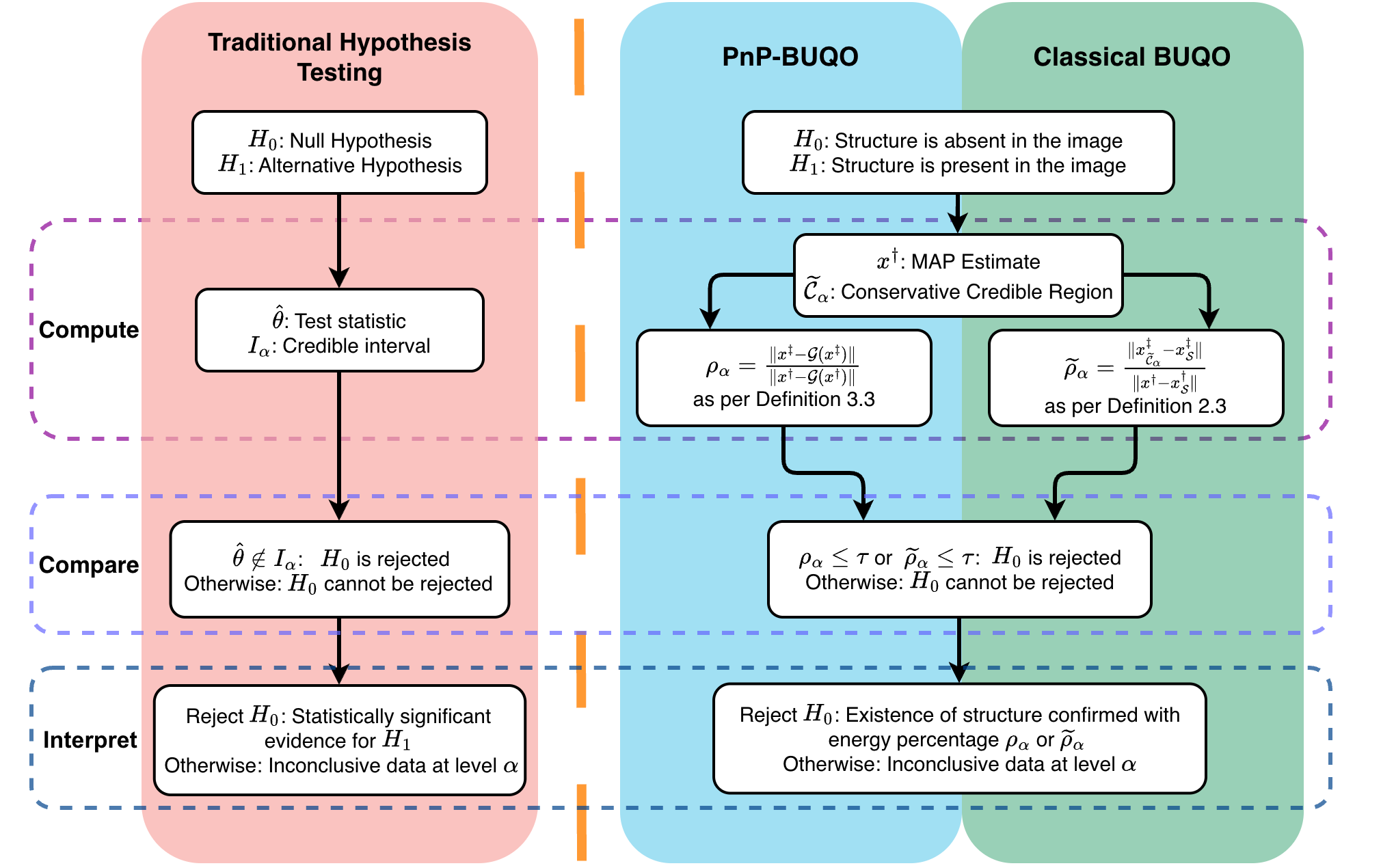}
\vspace*{+0.1in}
\caption{\small
A high-level overview of the workflow for the proposed PnP-BUQO method, comparing to the classical hypothesis testing, and the classical BUQO approach. }
\label{fig:comparison_pnp_buqo}
\end{figure*}

\section{Bayesian background and BUQO principle}\label{section:Bayes}


In this work, we focus on the case where a MAP estimate $\xb^\dagger$ is pre-computed for an inverse imaging problem of the form~\eqref{eq:model}, and we are unsure on whether some \textit{structure} appearing in $\xb^\dagger$ is true or fake, i.e., is supported by the data $\yb$, or is a reconstruction artefact (see an example in Figure~\ref{fig:pnp-buqo-demo}). 


Such a question can be addressed by performing a Bayesian hypothesis test, defined by postulating two hypotheses:
\begin{center}
$H_0$: Structure of interest is \textbf{ABSENT} \\ in the ground truth image;\\[0.1cm] 
$H_1$: Structure of interest is \textbf{PRESENT} \\ in the ground truth image.
\end{center}
Then, the test aims to use the posterior distribution to decide whether the null hypothesis $H_0$ can be rejected or not. Formally, according to Bayesian decision theory, $H_0$ is rejected (in favour of $H_1$) with significance $\alpha\in (0,1)$ if 
\begin{equation}\label{proba:H0-rej}
    \Ps (H_0 \mid \yb) \leq \alpha,
\end{equation}
in which case $ \Ps (H_1 \mid \yb) > 1-\alpha$. To compute the probability in \eqref{proba:H0-rej}, a standard method would be to rely on stochastic sampling (i.e., MCMC). 
Instead, the BUQO method \cite{RPW19} evaluates a Bayesian hypothesis test on the degree of support for specific image structures appearing in the MAP estimate $\xb^\dagger$, leveraging the scalability of convex optimization methods. 

\subsection{BUQO hypothesis testing}
\label{Ssec:BUQO-hyp}

In \cite{RPW19}, the null hypothesis $H_0$ is associated with a convex set $\Sc \subset \RR^N$ containing all possible images of $\RR^N$ \textbf{without} the structure of interest, i.e.,
\begin{equation}\label{proba:H0-rej-S}
    \Ps (H_0 | \yb ) = \Ps(\overline{\xb} \in \Sc \mid \yb) = \int_{\Sc} p(\xb | \yb) \ds \xb.
\end{equation}
A reminder of the hand-crafted definition of $\Sc$ proposed in \cite{RPW19} is provided in Appendix~\ref{appendix:struct}.
According to~\eqref{proba:H0-rej}, $H_0$ is rejected with significance $\alpha\in (0,1)$ if $\Ps(\overline{\xb} \in \Sc \mid \yb) \le \alpha$.
To evaluate this condition, while avoiding computing the high-dimensional integral in \eqref{proba:H0-rej-S}, BUQO aims to compare $\Sc$ with a posterior credible set, i.e., the region of the solution space where most of the posterior probability mass of $\xb|\yb$ lies \cite{Per17}. Formally, a set $\Cc_\alpha$ is a posterior credible region, with confidence level $(1-\alpha)$ for $\alpha\in (0,1)$, if $\Ps(\xb \in \Cc_\alpha|\yb) = 1-\alpha$. Since computing such a probability is highly expensive, authors in~\cite{RPW19} make use of a conservative credible region $\widetilde{\Cc}_\alpha$ introduced in \cite{Per17} and defined as
\begin{equation}    \label{def:Calpha}
    \widetilde{\Cc}_{\alpha}
    = \menge{\xb \in \RR^N}{ f_{\yb}( \Phib\xb) + g(\xb) \le \widetilde{\eta}_\alpha},
\end{equation}
where $f_{\yb}$ and $g$ are defined in~\eqref{eq:MAP}, and $\widetilde{\eta}_{\alpha} = f_{\yb}( \Phib\xb^\dagger) + g(\xb^\dagger) + N ( \sqrt{16 \log(3/\alpha)/N} + 1)$. It has been shown in \cite{Per17} that
\begin{equation*}
    \Ps(\overline{\xb} \in \widetilde{\Cc}_\alpha \mid \yb) \ge 1-\alpha,
\end{equation*}
and that $\widetilde{\Cc}_\alpha$ is the tightest approximation of the highest-posterior density region\footnote{The highest-posterior density region is defined as $\Cc_\alpha^* = \menge{\xb \in \RR^N}{ f_{\yb}( \Phib\xb) + g(\xb) \le {\eta}_\alpha^*}$, where $\eta_\alpha^* \in \RR$ is chosen such that $\int_{\Cc_\alpha^*} p(\xb|\yb \ds \xb = 1-\alpha$.} that can be obtained from the knowledge of the MAP estimate~\cite{Per17}. 


The BUQO approach makes use of the sets $\Sc$ and $\widetilde{\Cc}_\alpha$ to perform the hypothesis test and verify if $\Ps(\overline{\xb} \in \Sc \mid \yb) \le \alpha$ leveraging the following result, adapted from~\cite[Thm~3.2]{RPW19} and~\cite[Thm.~3.1]{Per17}.
\begin{theorem}\label{thm:1}
    Let $p(\xb|\yb)$ be defined as in \eqref{eq:post-def}, where $f_\yb$ and $g$ are assumed to be convex.
    Let $\alpha \in (4\exp(-N/3), 1)$. 
    If $\Sc \cap \widetilde{\Cc}_\alpha = \emp$, then $\Ps(H_0 | \yb)\leq \alpha$, and hence $H_0$ is rejected in favour of $H_1$. 
\end{theorem}

\begin{remark}
    It should be noted that the theorem above does not require convexity of set $\Sc$. In fact, to prove this result we only need to observe that
    (i) by definition of $\Sc$, we have $\Ps(\overline{\xb} \in \Sc | \yb) = \Ps (H_0 | \yb)$, 
    and (ii) if $\widetilde{\Cc}_\alpha \cap \Sc = \emp$ then $\Sc \subset \mathbb R^N \setminus \widetilde{\Cc}_\alpha$.
    Hence, combining (i) and (ii), we get 
    \begin{equation*}
        \Ps (H_0 | \yb) = \Ps(\overline{\xb} \in \Sc | \yb) = 1 - \Ps(\overline{\xb} \in \widetilde{\Cc}_\alpha | \yb).
    \end{equation*}
    All these facts are true for any subset $\Sc$ of $\mathbb R^N$ (without convexity assumption). The final result then follows from the fact that $\widetilde{\Cc}_\alpha$ is a convex conservative credible region, as $f_\yb$ and $g$ are both convex, using \cite[Thm.~3.1]{Per17}.
\end{remark}

According to Theorem~\ref{thm:1}, to perform the hypothesis test, it suffices to determine whether or not the sets $\tca$ and $\Sc$ are disjoint. If so, then all the images in $\tca$ contain the structure of interest and $H_0$ is rejected in favor of $H_1$ with confidence $\alpha$. Otherwise, there exists an estimate $\xb$ which satisfies the observation $\yb$ and is structure-free, so we may not draw any conclusions with respect to the hypothesis test. It is to be noted that, it is not sufficient to have $\Sc \cap \tca \neq \emp$ to accept the hypothesis $H_0$, since there are images containing the structure of interest and others that do not, all of them being supported by the measured data.

\subsection{Variational problem}
\label{Ssec:BUQO-var-form}

In this section we recall the BUQO method proposed in~\cite{RPW18, RPW19} to evaluate whether the intersection $\Sc \cap \tca$ is empty. 

We first need to introduce some notation that will be used to provide a mathematical definition of set $\Sc$. Suppose we have identified an area of the MAP estimate (i.e., a collection of pixels denoted $\II_\Mb \subset \{1, \ldots, N\}$) which comprise a structure of interest.
Let $\xb \in \RR^N$.
Let $N_\Mb = \Card{\II_\Mb}$, and $\Mb\in \{0,1\}^{N_\Mb \times N}$ be the linear masking operator such that $\Mb\xb = (\xb_j)_{j\in \II_\Mb} \in \RR^{N_\Mb}$ is the restriction of $\xb$ to the structure of interest.  Further let $\Mb^c \in \{0,1\}^{(N-N_\Mb)\times N}$ be the complementary mask so that $\Mb^c \xb = (\xb_j)_{j\not\in \II_\Mb} \in \RR^{(N-N_\Mb)}$ is a vector whose entries are not part of the structure of interest.

\paragraph{Hypothesis test as a minimization problem}
The authors in \cite{RPW19} propose to reformulate the problem of determining if $\Sc \cap \tca = \emp$ as measuring the distance between $\Sc$ and $\tca$. If this distance is not $0$, then the intersection is empty, otherwise it is not.
Measuring the distance between $\Sc$ and $\tca$ can be written as a minimization problem:
\begin{equation}\label{eq:min}
    \minimize{(\xb_{\Sc}, \xb_{\widetilde{\Cc}_\alpha}) \in \RR^{2N}} \frac \gamma 2 \| \xb_{\Sc} - \xb_{\widetilde{\Cc}_\alpha}\|_2^2 +\iota_{\Sc}(\xb_{\Sc}) + \iota_{\widetilde{\Cc}_\alpha} (\xb_{\widetilde{\Cc}_\alpha}),
\end{equation}
where $\gamma>0$ and $\iota_\Sc$ denotes the indicator function of set $\Sc$, which is equal to $0$ if its argument belongs to $\Sc$, and $+\infty$ otherwise. In other words, we want to find two images, $\xb_{\Sc}$ in $\Sc$ and $\xb_{\widetilde{\Cc}_\alpha}$ in $\widetilde{\Cc}_\alpha$, such that the $\ell_2$ norm of the difference between the two is minimal.
In \eqref{eq:min}, $\gamma>0$ is a free parameter not impacting the solution, but used in practice for acceleration purpose (see \cite{RPW18} for details).

From Theorem~\ref{thm:1}, we deduce the following:
\begin{corollary}\label{cor:1}
    Suppose that $\alpha \in ]4\exp(-N/3), 1[$ and let $(\xb_{\Sc}^\ddagger, \xb_{\tca}^\ddagger) \in \Sc \times \widetilde{\Cc}_\alpha$ be a solution to~\eqref{eq:min}. 
    If $\| \xb_{\Sc}^\ddagger - \xb_{\widetilde{\Cc}_\alpha}^\ddagger\| > 0$, then $\Ps(H_0 \mid \yb)\leq \alpha$, and hence $H_0$ is rejected. 
\end{corollary}

Note that, in practice, achieving $\| \xb_{\Sc}^\ddagger - \xb_{\widetilde{\Cc}_\alpha}^\ddagger\| = 0$ is very difficult due to numerical approximations when solving~\eqref{eq:min} iteratively. Instead, a threshold value can be fixed, below which it can be assumed that $\| \xb_{\Sc}^\ddagger - \xb_{\widetilde{\Cc}_\alpha}^\ddagger\| = 0$. In practice this value should be chosen small enough looking at the energy of the structure in the MAP estimate $\xb^\dagger$ \cite{RPW18, PNAS2022}. To this aim, the authors in \cite{RPW18} introduced the parameter below that acts similarly as a $p$-value in a standard hypothesis testing setting.

\begin{definition}\label{def:rho_tilde}
Let \(\xb^\dagger \in \RR^N\) be a MAP estimate and $\xb^\dagger_\Sc \in \Sc$ be a structure-free version of $\xb^\dagger$. Let $(\xb_{\Sc}^\ddagger, \xb_{\tca}^\ddagger) \in \Sc \times \widetilde{\Cc}_\alpha$ be a solution to~\eqref{eq:min}. The normalized intensity of the structure in $\xb_{\widetilde{\Cc}_\alpha}^\ddagger$ is defined as
\begin{equation*}
    \widetilde{\rho}_\alpha = \frac{\|  \xb_{\widetilde{\Cc}_\alpha}^\ddagger - \xb_{\Sc}^\ddagger\|}{\| \xb^\dagger - \xb^\dagger_\Sc \|} \in [0,1].
\end{equation*}
\end{definition}
Following Corollary~\ref{cor:1}, the larger is $\widetilde{\rho}_\alpha$, and the more confident we are in rejecting the null hypothesis $H_0$. This value can also be interpreted as an upper bound on the percentage of energy of the structure that is supported by the data. 

\paragraph{BUQO implementation}

BUQO performs hypothesis tests by solving problem~\eqref{eq:min}. To solve such a problem, one can use proximal splitting algorithms \cite{combettes2011proximal, CombettesPesquet2021, BauschkeCombettes2011, chambolle2016introduction}, in particular primal-dual methods~\cite{Komodakis2014}. To this aim, it is crucial define $\Sc$ such that the projection onto this set is feasible. 

In previous works \cite{RPW19, PNAS2022}, the set $\Sc$ was defined as an intersection of convex sets, i.e., $\Sc = \cap_{s=1}^{S} \Sc_s$, where, for every $s \in \{1, \ldots, S\}$, $\Sc_s$ is convex, closed and proper.
The choice of sets $(\Sc_s)_{1 \le s \le S}$ then depends on the type of structure for which one wants to perform the hypothesis test. In general, one of the sets aims to \textit{remove} the structure, another to \textit{inpaint}, and other sets can be added for physical constraints (e.g.,  positivity constraint for intensity images). 
Examples are provided in Appendix~\ref{appendix:struct}.

This definition of $\Sc$ suffers a number of drawbacks. In particular, there are multiple hyper-parameters that must be manually tuned by the user. These parameters can be seen as tolerances the user allow on how \textit{smooth} they want the inpainting to be, or how much energy can be considered to suppose the structure is absent (i.e., to \textit{remove} the tructure). Empirical choices of these parameters are proposed in \cite{RPW18, RPW19, PNAS2022}, based on statistics computed on the MAP estimate. However this approach lacks generalization. In addition, some of the operators used in previous works are application-specific, and/or image-dependent (see Appendix~\ref{appendix:struct}), further emphasizing the need for more generic definitions of $\Sc$.



\section{Inpainting-PnP primal-dual method for hypothesis testing}\label{section:problem}

Building on the classical BUQO formulation and motivated by the limitations of hand-crafted inpainting, we are now ready to propose a fully data-driven BUQO framework. In Figure~\ref{fig:comparison_pnp_buqo}, we provide a high-level overview of the workflow for the proposed PnP-BUQO, explaining differences and similarities compared to the classical BUQO approach, as well as to classical hypothesis testing.

\subsection{Proposed fixed-point formulation}
\label{Ssec:BUQO-DL-form}

A major contribution of this work is the proposal of a data-driven definition for structure-free images in terms of a (not necessarily linear) inpainting operator. Hence, we first define such an inpainting operator.
\begin{definition} \label{def:inp-op}
Let $\xb \in \eR^N$ be an image, and $\Mb \in \{0,1\}^N$ be a binary mask identifying a localized structure in $\xb$. Then
\begin{enumerate}
    \item 
    $\Gc \colon \eR^N \times \{0,1\}^N \rightarrow \eR^N$ is an inpainting operator for $\Mb$ if $\Mb^c \Gc (\xb, \Mb) = \Mb^c \xb$, 
    \item
    we say that $\xb$ is \emph{structure-free} (with respect to $\Mb$) if it is a fixed point of $\Gc(\cdot, \Mb)$, i.e., if $\Gc(\xb, \Mb) = \xb$. 
\end{enumerate}
\end{definition}

To simplify notation, in the following we omit $\Mb$ when using $\Gc$. Therefore we define $\Sc$ to be the set of fixed points of $\Gc$ (with respect to $\Mb$), by
\begin{equation}\label{eq:Sc}
\Sc = \menge{\xb \in C }{ \xb \approx \Gc(\xb) },
\end{equation}
where $C \subset \RR^N$ is a feasible set for image $\xb$, e.g., corresponding to $[0,+\infty)^N$ for intensity images, or $[0,1]^N$ for normalized images. In the remainder we assume that $\tca \subset C$, that is, the MAP estimate is feasible. The approximation in~\eqref{eq:Sc} is necessary to ensure that the probability of $\Sc$ is non-zero.

While $\Gc$ is stated in terms of an arbitrary inpainting operator, NNs form a canonical family of operators from which we can choose a particular application-specific $\Gc$. This generality allows the expert practitioner to be consulted in crafting a bespoke NN, or transferring an existing model, to generate meaningful structure-free images. Although we will not focus on a particular definition for $\Gc$ in this and the following section, we will propose an inpainting convolutional neural network (CNN) in Section~\ref{Sec:inp-cnn}.

Using \eqref{eq:Sc}, Definition~\ref{def:inp-op}, and building on BUQO, we propose to perform the hypothesis test by solving the minimization problem
\begin{equation}\label{eq:BUQO_Fixed}
\minimize{\xb \in C}
\dfrac{\zeta}{2}\| \xb - \Gc (\xb)\|_2^2
+ \iota_{\tca}(\xb) ,
\end{equation}
where $C \subset \RR^N$.
Solving \eqref{eq:BUQO_Fixed} allows us to find an image $\xb^\ddagger \in \RR^N$ that belongs to $\tca$. Then, according to~\eqref{eq:Sc}, if $\| \xb^\ddagger - \Gc (\xb^\ddagger)\|_2$ is sufficiently small, it means that $\xb^\ddagger$ also belongs to $\Sc$. 
Thus, the following result can be deduced from Theorem~\ref{thm:1}.
\begin{corollary}\label{cor:2}
Let $\Sc = \menge{\xb \in C}{ \xb = \Gc(\xb)}$.
    Suppose that $\alpha \in ]4\exp(-N/3), 1[$ and let $\xb^\ddagger \in \tca $ be a solution to~\eqref{eq:BUQO_Fixed}. 
    If $\| \xb^\ddagger - \Gc (\xb^\ddagger)\| > 0$, then $\Ps(H_0 \mid \yb)\leq \alpha$, and $H_0$ is rejected. 
\end{corollary}
\begin{proof}
According to \eqref{eq:BUQO_Fixed}, for every $\xb \in \tca$, we have $\| \xb^\ddagger - \Gc (\xb^\ddagger)\|_2 \le \| \xb - \Gc (\xb)\|_2 $. So, if $\| \xb^\ddagger - \Gc (\xb^\ddagger)\|_2 > 0$, then there is no image in $\tca$ that is a fixed point of $\Gc$. Hence, by definition of $\Sc$, we can deduce that $\tca \cap \Sc = \emp$.
\end{proof}

\subsubsection*{Remarks}
Unlike the original formulation described in Section~\ref{Ssec:BUQO-var-form}, the proposed formulation does not look for the distance between set $\tca$ and set $\Sc$. Instead, the solution to~\eqref{eq:BUQO_Fixed} provides an upper bound on this distance, and a lower bound on the distance between the MAP estimate $\xb^\dagger$ and its structure-free version, i.e., 
\begin{equation*}
\| \xb^\dagger - \Gc(\xb^\dagger) \|_2 \ge \| \xb^\ddagger - \Gc(\xb^\ddagger) \|_2 \ge \text{dist}(\Sc, \tca).    
\end{equation*}
The proposed definition of $\Sc$ could also be used in a similar way as in problem~\eqref{eq:min}, e.g. by solving
\begin{equation*}
    \minimize{\xb_\Sc \in C, \, \xb_\tca \in \tca} 
\dfrac{\gamma}{2}\| \xb_\tca -  \xb_\Sc\|_2^2
+ \dfrac{\zeta}{2} \| \xb_\Sc - \Gc(\xb_\Sc) \|_2^2.
\end{equation*}
Instead, problem~\eqref{eq:BUQO_Fixed} directly determines whether the intersection $\tca \cap \Sc$ is empty, which is sufficient for the hypothesis test, while reducing the dimension by looking for a single image.

In addition, the proposed formulation of structure-free images, is advantageous over the hand-crafted definition described in Appendix~\ref{appendix:struct}, since it is parameter-free. In addition, when the model has been trained for a particular application and structure type (e.g., localized structures), operator $\Gc$ can be used for any mask without needing to be adapted. Further, as we will show in the simulation section, the results from the proposed approach outperform results of BUQO in terms of visual inspection.

To evaluate our hypothesis test using Corollary~\ref{cor:2}, it suffices to determine whether or not the sets $\Sc$ and $\tca$ are disjoint. Similar to described in Section~\ref{Ssec:BUQO-var-form}, instead of determining if $\| \xb^\ddagger - \Gc(\xb^\ddagger) \| >0$, we introduce a threshold value below which we assume $\| \xb^\ddagger - \Gc (\xb^\ddagger)\| =0$. 
This value mainly accounts for numerical approximation limitations in algorithmic implementations. 
Hence, although the definition of $\Sc$ is parameter-free, the hypothesis test is dependent on this threshold value. 

\begin{definition}
Let \(\xb^\dagger\) be a MAP estimate,  $\xb^\ddagger \in \tca $ be a solution to~\eqref{eq:BUQO_Fixed} and let $\Gc \colon \eR^N \times \{0,1\}^N \rightarrow \eR^N$ be an inpainting operator for $\Mb$ as per Definition~\ref{def:inp-op}. The normalized intensity of the structure is defined as
\begin{equation*}
    \rho_\alpha = \frac{\| \xb^\ddagger - \Gc(\xb^\ddagger)\|}{\|\xb^\dagger - \Gc(\xb^\dagger)\|} \in [0,1].
\end{equation*}
\end{definition}
The quantity $\rho_\alpha$ has the same interpretation as $\widetilde{\rho}_\alpha$ in Definition~\ref{def:rho_tilde}, i.e., it acts as an upper bound on the percentage of the energy of the structure that is supported by the data. The greater the value of $\rho_\alpha$, the more confidence we have in rejecting the null hypothesis. Hence $\rho_\alpha$ acts similarly to a p-value in classical hypothesis testing. To decide whether the null hypothesis $H_0$ can be rejected or not, we consider a threshold value $ 0 < \tau \ll 1 $. If $\rho_\alpha > \tau$ then $\Sc \cap \tca = \emp$ and the null hypothesis $H_0$ can be rejected with confidence $\alpha$. Otherwise, $\rho_\alpha \leq \tau$ and we conclude that there is a structure-free image that is concurrent with the observed data. Hence the data is inconclusive and the null hypothesis cannot be rejected. 

\begin{remark} Note that problem~\eqref{eq:BUQO_Fixed} is usually not convex. In this case, the output of the algorithm would therefore be a critical point $x^{\infty}$ rather than the solution $x^\ddagger$ to the minimization problem. Hence, we would have $\| x^\ddagger - \mathcal{G}(x^\ddagger) \| \le \| x^\infty - \mathcal{G}(x^\infty) \| $. In this case, we make the following comments. 
\begin{enumerate}
        \item 
        The percentage of the structure that will be confirmed by the data looking at $x^\infty$ will be overestimated. In fact, we have
        \begin{equation*}
            \rho_\alpha = \frac{\| x^\ddagger - \mathcal{G}(x^\ddagger) \|}{\| x^\dagger - \mathcal{G}(x^\dagger) \|}
            \le 
            \rho_\alpha^\infty = \frac{\| x^\infty - \mathcal{G}(x^\infty) \|}{\| x^\dagger - \mathcal{G}(x^\dagger) \|}.
        \end{equation*}
        \item 
        The fact of obtaining a critical point instead of a minimizer is conservative regarding the uncertainty with respect to the structure of interest.
        Let $\tau>0$ be the threshold to decide whether $H_0$ should be rejected or not.
        \begin{itemize}
            \item 
            If $\rho_\alpha \le \rho_\alpha^\infty \le \tau$, then $H_0$ is rejected, confirming the structure is in the true image. This conclusion is the same as the one that should be drawn if $x^\ddagger$ was available.
            \item 
            If $ \rho_\alpha \le \tau \le \rho_\alpha^\infty $, then $H_0$ cannot be rejected based on $x^\infty$, while it would be rejected if $x^\ddagger$ was available. In this case, it means that $x^\infty$ does not give enough evidence against $H_0$ and that we are still unsure about the structure.
        \end{itemize}
    \end{enumerate}
\end{remark}




\subsection{Proposed PnP-BUQO}\label{sect:algos}

In this section, we present the proposed PnP-BUQO framework for performing the hypothesis test with the data-driven structure definition, as described in Section~\ref{Ssec:BUQO-DL-form}. Before introducing algorithmic details of our proposed method, we describe the MAP problem of interest.

\subsubsection{MAP estimation}\label{subsec:map_estimate}
For the sake of simplicity, in this work we focus on the same setting as in \cite{RPW18}, assuming that the additive noise in~\eqref{eq:model} has a bounded energy; that is, there exists $\varepsilon>0$ such that $\|\wb\|_2 = \| \Phib \xb - \yb \|_2 \leq \varepsilon$. 
Equivalently, this constraint can be reformulated as $\Phib \xb \in \Bc_2(\yb, \varepsilon)$, where $\Bc_2(\yb, \varepsilon)$ denotes the $\ell_2$-ball, centred in $\yb$ with radius $\varepsilon$.
Considering problem~\eqref{eq:MAP}, we thus set our data fidelity term as 
\begin{equation}    \label{eq:fy-def}
    f_\yb(\Phib \xb) = \iota_{\Bc_2(\yb, \varepsilon)}(\Phib \xb).
\end{equation}
Furthermore, we choose a hybrid regularization term, imposing a constraint on the dynamic range of the image of interest, and promoting sparsity in a transformed domain. The resulting function is of the form of
\begin{equation}    \label{eq:g-def}
    g(\xb) = \iota_{[0,1]^N}(\xb) + \lambda \|\Psib \xb \|_1,    
\end{equation}
where $\lambda>0$ is a regularization parameter, and $\Psib \colon \eR^N \to \eR^P$ models a linear sparsifying operator. Common choices of $\Psib$ include wavelet transforms~\cite{Mallat, Daubechies1988, Daubechies1992, daubechies1992ten}, or gradient differences (e.g. vertical/horizontal differences, resulting in the anisotropic TV norm \cite{ROF1992}).

According to the definition of $\tca$ in Eq.~\eqref{def:Calpha}, when functions $f_\yb$ and $g$ are specified by \eqref{eq:fy-def} and \eqref{eq:g-def}, respectively, we have
\begin{multline}    \label{eq:Calpha-balls}
    \tca = \Big\{\xb \in [0,1]^N \, \mid \, \Phib \xb \in \Bc_2(\yb, \varepsilon) \\[-0.2cm]
    \text{ and } \Psib \xb \in \Bc_1(\zerob, \widetilde{\eta}_\alpha/\lambda) \Big\},
\end{multline}
where $\Bc_1(\zerob, \widetilde{\eta}_\alpha/\lambda)$ denotes the $\ell_1$-ball, centered in $\zerob$ with radius $\widetilde{\eta}_\alpha/\lambda$. 

We proceed to solve the resulting particular instance of problem~\eqref{eq:MAP} with the primal-dual Condat-V\~u method \cite{Condat13, Vu13} described in Algorithm~\ref{alg:PD_MAP} in Appendix~\ref{appendix:map}. In the remainder, we denote by $\xb^\dagger$ the output of this method. 

For the sake of completeness, we also provide in Appendix~\ref{appendix:buqo} the implementation of the classical BUQO, when $\Sc = \Sc_1 \cap \Sc_2 \cap \Sc_3$ as explained in Section~\ref{Ssec:BUQO-var-form}.



\begin{figure*}[t]
\centering
\begin{overpic}[trim = {0.35in, 0.4in, 0.5in, 0.4in}, clip, width = 0.75\textwidth]{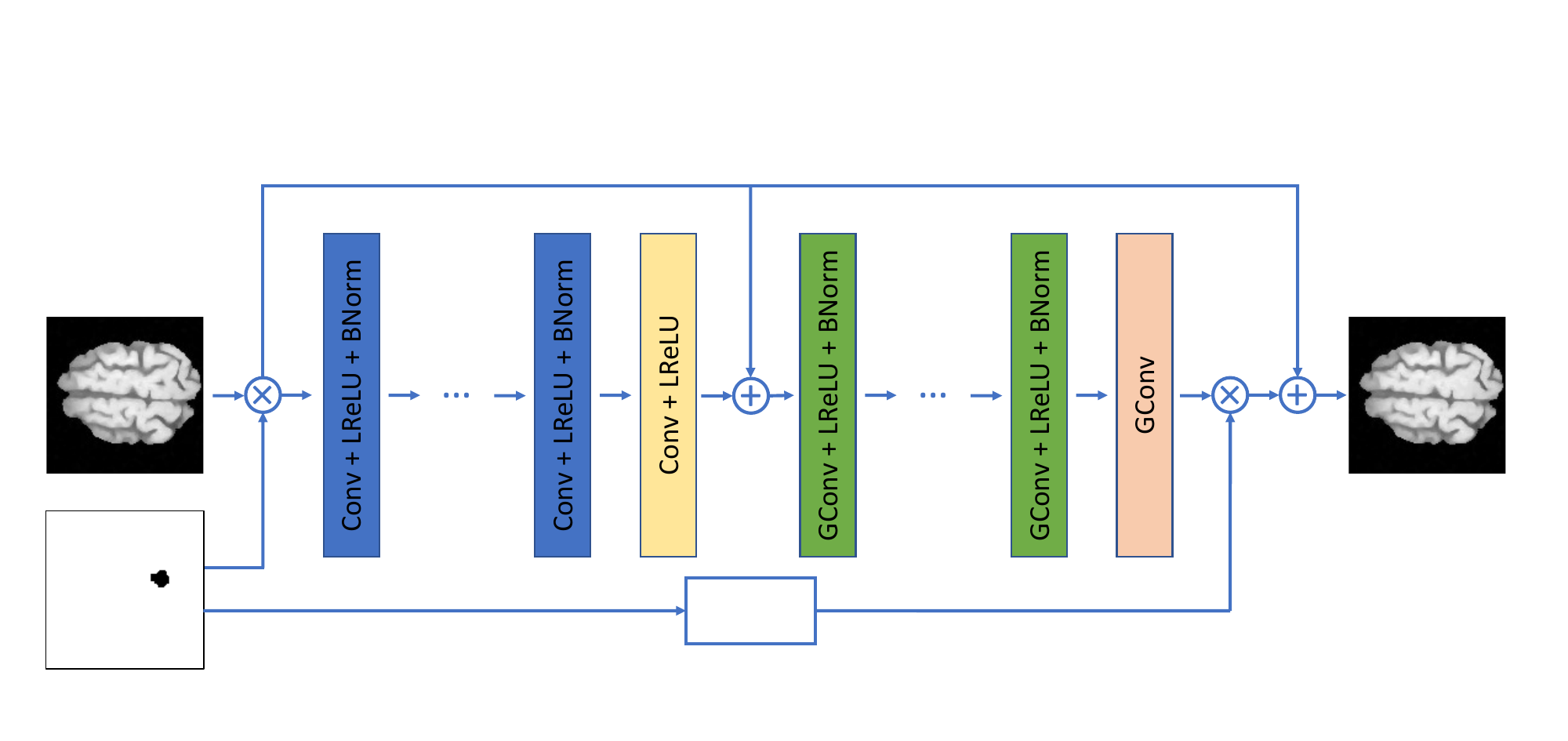} \put (1,3) {\small $\Mb$}
\put (44.4, 4.7) {\small$1 - \Mb$}
\put (1, 16) {\small \color{white} $\xb$}
\put(89.5, 16) {\small \color{white}  $\Gc(\xb)$}
\end{overpic}


\caption{\small
Modified DnCNN architecture, where Conv is a $2d$ convolutional layer, LReLU is a leaky rectified linear unity, BNorm is a batch normalization layer and GConv is a $2d$ gated convolutional layer.}
\label{fig:myGatedDnCNN}
\end{figure*}

\subsubsection{PnP primal-dual algorithm for BUQO}

In this section, we detail the algorithmic realization of our PnP version of BUQO approach, using a primal-dual forward-backward scheme to solve problem~\eqref{eq:BUQO_Fixed}.
As suggested in Section~\ref{Ssec:BUQO-DL-form}, CNNs have shown outstanding performance in inpainting tasks (see, e.g., \cite{pconv, Yu19, ShiftNet, palette2022, wang2023ddnm}). We hence propose to use a CNN for the operator $\Gc$ appearing in the definition \eqref{eq:Sc} of set $\Sc$. Further details on the choice of $\Gc$ are provided in Section~\ref{Sec:inp-cnn}.

Let $\Gc$ be a Lipschitz differentiable NN trained to solve an inpainting task, and let
\begin{equation*}
    (\forall \xb \in \RR^N)\quad
    h(\xb) = \dfrac{\zeta}{2} \|\xb - \Gc(\xb)\|_2^2.
\end{equation*}
We can solve the minimization problem \eqref{eq:BUQO_Fixed} by computing gradient steps on the Lipschitz-differentiable function $h$, and projections on sets $\tca$ and $C$. 
In particular, one can evaluate $\nabla h$ by automatic differentiation. In the remainder, we denote by $\beta>0$ the Lipschitz constant of $\nabla h$.

Due to the non-linearity of operator $\Gc$, function $h$ is not convex. In this context, problem~\eqref{eq:BUQO_Fixed} can be solved with forward-backward methods \cite{Chouzenoux2014, Attouch11} that would provide a critical point of the objective function. However, since the set $\tca$ defined in \eqref{eq:Calpha-balls} corresponds to the intersection of three convex sets, the projection onto this set would require sub-iterations. Instead, we propose to use a primal-dual algorithm. Although this algorithm does not have convergence guarantees for non-convex objectives, it has been used in the literature in this context \cite{ProceedingsCondat13}. The resulting iterations are provided in Algorithm~\ref{alg:PD_Fix}.
Since the definition of structure-free images is now entirely characterized by the function $h$, Algorithm~\ref{alg:PD_Fix} involves fewer steps than the original BUQO Algorithm (see Appendix~\ref{appendix:buqo}).

Additionally, in contrast to the original BUQO Algorithm, structure-removal updates in Algorithm~\ref{alg:PD_Fix} occur via the gradient step, and is updated by $\nabla h$ in modifying the variable $\xb$. Since the effective domain of $\Gc$ is the collection of pixels associated to $\Mb^c$, the gradient term corresponding to inpainting would update precisely those pixels that are not part of the structure of interest. 

The output of Algorithm~\ref{alg:PD_Fix} is $\xb^\ddagger\in \tca$, which aims to satisfy $ \xb^\ddagger = \Gc(\xb^\ddagger) $. Note that $\xb^\ddagger$ is analogous to $\xb^\ddagger_{\tca}$ given by the original BUQO in Corollary~\ref{cor:1}, while $\Gc(\xb^\ddagger)$ is analogous to $\xb^\ddagger_\Sc$. 


Although convergence guarantees for Algorithm~\ref{alg:PD_Fix} only hold for convex functions, we still choose the step-sizes $(\mu_1, \mu_2, \sigma) \in [0,+\infty)^3$ as per \cite[Thm. 5.1]{Condat13}, i.e., to satisfy
\begin{equation}\label{eq:stepsizes} 
\frac 1 \sigma \geq \dfrac{\beta }{2} +  \mu_1 \| \Psib \|_S^2 + \mu_2  \| \Phib \|_S^2.
\end{equation}

Note that, due to the fact that $\Gc$ is chosen to be a NN, the proposed Algorithm~\ref{alg:PD_Fix} is reminiscent of PnP methods. In contrast to standard PnP methods using NN as a regularizer, we propose to use a NN for an inpainting procedure. While both PnP algorithms and the use of NNs for image inpainting are ubiquitous, to the best of our knowledge, there are no PnP algorithms in the literature incorporating inpainting NNs. Thus, Algorithm~\ref{alg:PD_Fix} is novel in this regard.

\begin{remark}
Note that if $\Gc$ is chosen such that $\nabla h$ is ensured to be monotone, then iterations from Algorithm~\ref{alg:PD_Fix} converge to a solution to problem~\eqref{eq:BUQO_Fixed} \cite{Condat13, Vu13}. 
Such a condition could possibly satisfied in practice if using training framework similar to the one proposed in~\cite{belkouchi2025learning}.
\end{remark}

\begin{algorithm}[t]
\caption{PnP-BUQO algorithm}\label{alg:PD_Fix}
\begin{algorithmic}
\STATE{\textbf{Initialization:} 
Let $\xb^{(0)} \in ([0,\infty[^{N})^2$ and 
$(\vb_1^{(0)}, \vb_2^{(0)}) \in \eR^P \times \eC^M$. 
Let $(\mu_{1}, \mu_{2}, \sigma) \in [0,\infty)^3$ be such that \eqref{eq:stepsizes} is satisfied. 
}
\STATE{\textbf{Iterations:}}
 \FOR{$k=0, 1, \ldots$}
  \STATE{$\widetilde{\vb}_1^{(k)} = \vb_1^{(k)} + \mu_1 \Psib \xb^{(k)}$}
  \STATE{${\vb_1}^{(k+1)} = \widetilde{\vb}_1^{(k)} - \mu_1\Pi_{\Bc_1(0, \widetilde{\eta}_\alpha/\lambda)} \Big( \mu_1^{-1}\widetilde{\vb}_1^{(k)} \Big) $}
  \STATE{$\widetilde{\vb}_2^{(k)} = \vb_2^{(k)} + \mu_2 \Phib \xb^{(k)}$}
  \STATE{$\vb_2^{(k+1)} = \widetilde{\vb}_2^{(k)} - \mu_2 \Pi_{\Bc_2(\yb, \varepsilon)} \Big( \mu_2^{-1}\widetilde{\vb}_2^{(k)} \Big) $}
  \STATE{$ \!\!\! \begin{array}{ll}
                \widetilde\xb^{(k)} \!\!\!\!  &= \Pi_{C}\bigg(\xb^{(k)}- \sigma \nabla h(\xb^{(k)})  \\
                                    &   \quad \quad \quad \quad \quad \quad -\sigma \Psib^\dagger \vb_1^{(k+1)} -\sigma  \Phib^\dagger \vb_2^{(k+1)} \bigg)
            \end{array}$}
  \STATE{$\xb^{(k+1)} = 2\widetilde\xb^{(k)} - \xb^{(k)}$}
 \ENDFOR
\end{algorithmic}
\end{algorithm}


\section{Numerical simulation setting}\label{section:simulations}

All our experiments are coded with Pytorch, and performed on an NVIDIA GeForce RTX 2080 Ti.
To evaluate the proposed uncertainty quantification framework, we use \eqref{eq:model} as a simplified model of Fourier Magnetic Resonance (MR) imaging as well as of the Computed Tomography (CT) imaging. In particular, for the MR imaging tasks we consider the pre-processed and annotated real-valued images from the BRATS21 dataset \cite{BRATS1, BRATS3, BRATS2}, cropped to dimension $N = 128 \times 128$. For the CT imaging tasks, we use the standard Shepp-Logan phantom, rescaled to dimensions \(120 \times 120\), as well as realistic CT scans from the Cancer Genome Atlas Lung Adenocarcinoma (TCGA-LUAD) data collection \cite{albertina2016cancer}, of dimension $N = 512 \times 512$.


\subsection{An Inpainting CNN}
\label{Sec:inp-cnn}

\paragraph{Proposed network architecture}

A wealth of architectures and training regimes are available in the literature to perform image inpainting~\cite{pconv, ShiftNet, Yu19}. In this work, due to the form of Algorithm~\ref{alg:PD_Fix} that requires to differentiate $h$ (and hence the NN $\Gc$) at each iteration, we propose to use a CNN with fairly simple architecture. In particular, we choose $\Gc$ to be a modified DnCNN, which contains both convolutions and gated convolutions~\cite{Yu19, ZZCMZ17}. 

Gated convolutions are modified convolutional layers allowing for inpainting holes of any shape and without specifying those pixels to be inpainted \cite{Yu19}. While a standard convolution is computed by weighting all pixels in the receptive field equally, a gated convolution is a convolution that is modulated by non-uniform pixel weightings, via a learnable gating mechanism. Specifically, for each gated convolution, depending on both the spatial location and feature channel of the input, a gating map with values in $[0,1]$ is learned. This gating map is used to adaptively scale the convolution outputs, reflecting how much each pixel should be trusted.


The proposed modified DnCNN network $\Gc$ takes as an input both the mask $\Mb$ and an image $\xb$ to be inpainted, in the area identified by the mask. In particular, we want $\Gc$ to satisfy Definition~\ref{def:inp-op}, so the output image $ \Gc(\xb)$ corresponds to an inpainted version of $\Mb \xb$ and should satisfy $\Mb^c \Gc (\xb) = \Mb^c \xb$.
The proposed architecture, depicted in Figure~\ref{fig:myGatedDnCNN}, is split into two blocks: a first denoising block to improve robustness, and a second block aiming to perform the inpainting operation. The initial denoising block contains $5$ standard convolutional layers, each with $128 \times 5\times 5$ filters. The second inpainting block contains $5$ gated convolutional layers, with $128 \times 5\times 5$ filters.

\paragraph{Training procedure}

For training the network, we adopt a similar approach to~\cite{pconv}. The network $\Gc$ is trained by minimizing a composite loss function comprising of five terms: an $\ell^2$-loss, an $\ell^1$ loss restricted to the inpainted region, a TV term restricted to the mask boundary and a perceptual and style loss using the VGG16 feature extractor~\cite{VGG16}. 
Furthermore, we use activation dropouts on the first four layers. We use a dataset containing $90,885$ images and masks from the BRATS21 dataset \cite{BRATS1, BRATS3, BRATS2}. Images are preprocessed with additive random noise and random rotations. Masks are preprocessed from the label data, and are randomly dilated. Multiple masks are multiplied to form a new single mask. The processed images are randomly cropped to $96 \times 96$ and masked using a randomly chosen mask from our mask dataset. The model is trained with the ADAM optimizer, for 32 epochs with a batchsize of 24. The learning rate is $10^{-3}$, dropping to $2\times 10^{-4}$ after 16 epochs. Training was performed with PyTorch, on an NVIDIA GeForce RTX 2080 Ti and lasted 16~hours.

\subsection{Algorithm parameters}\label{subsect:params}

To implement Algorithm~\ref{alg:PD_Fix}, the stepsizes must satisfy \eqref{eq:stepsizes}, that necessitates to compute the Lipschitz constant $\beta$ of $\nabla h$. To do so, we approximate it as $\beta \approx \max_{1 \le i \le I} \| \nabla^2 h(\Gc(\xb^\dagger) + \omegab_i)\|_S$, where $(\omegab_i)_{1 \le i \le I}$ are realizations of random Gaussian variables with standard deviation $\varsigma = 0.01$. In practice, we chose $I=4$, and we compute $\nabla^2 h(\xb^\dagger)$ by automatic differentiation.
We then compute $\|\nabla^2 h(\Gc(\xb^\dagger) + \omegab_i)\|_S^2$ by power iterations~\cite{PRTW21}. 

Note that, according to Definition~\ref{def:inp-op}, for every $\xb \in \RR^N$, $\Gc(\Gc(\xb)) = \Gc(\xb)$. That is, $h(\Gc(\xb)) = 0$, and hence $\nabla h(\Gc(\xb)) = 0$. 
Since we propose to define $\Sc$ using a CNN, we may have poor control over $\nabla h$ in general. However, since we aim to find an element of $\Sc$, if we initialize Algorithm~\ref{alg:PD_Fix} with $\xb^{(0)} = \Gc(\xb^\dagger) \in \Sc$, 
and choose a sufficiently small primal stepsize $\sigma$, we aim to keep our primal steps small, so that each $\xb^{(k)}$ is sufficiently similar to the training dataset. 
In fact, preliminary simulations suggested that this local approximation of $\beta$ is good enough to ensure stability of the proposed PnP-BUQO iterations.
The behaviour of $(h(\xb^{(k)}))_k$ and $\| \nabla h(\xb^{(k)}) \|_S$ will be investigated in Section~\ref{Ssec:res:behav}. We consider that the algorithm has converged when
\begin{equation}   \label{eq:stop-cond}
 \xb^{(k)} \in \tca 
 \text{ and }
\| \xb^{(k+1)} - \xb^{(k)} \| \le 10^{-3} \| \xb^{(k)} \|,
\end{equation}
where $(\xb^{(k)})_{k \in \NN}$ are the iterates generated by Algorithm~\ref{alg:PD_Fix}.

For the feasibility constraint $C$ representing the dynamic of the image, we choose $[0,1]^N$ due to our choice of application.

For the hypothesis test evaluation in our simulations, we choose $\alpha = 0.01$ and $ \tau = 0.02$.

\begin{figure}[ht]
    \centering
    \includegraphics[width = 0.35\textwidth]{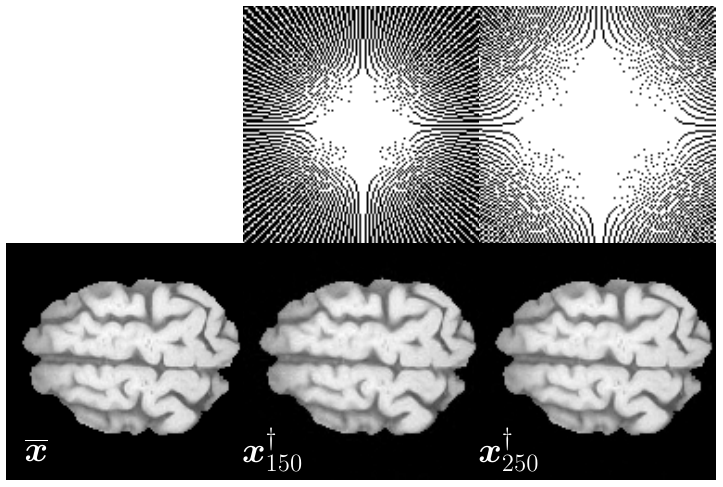}
    \caption{\small
    \textbf{Discrete Fourier sampling.} (top) Examples of masks $\Mb$ with $M_a = 150$ and $M_a = 250$ angles (i.e., $M/N = 0.52$ and $M/N = 0.76$), respectively. (bottom) Ground truth ($N = 128 \times 128$), and MAP estimates for  $\mathrm{iSNR} =30$~dB.}
    \label{fig:masks_map}
\end{figure}

\subsection{Simulated data models}\label{subsect:sim}

We consider three measurement models for operator $\Phib$ in \eqref{eq:model}, described below. The associated MAP estimates $\xb^\dagger$, which is a solution to~\eqref{eq:MAP}, are obtained using Algorithm~\ref{alg:PD_MAP}, where the sparsifying operator $\Psib$ is chosen to be the Daubechies wavelet transform (db8) with decomposition level $3$.

\paragraph{Discrete Fourier sampling} 

We simulate measurement data by considering the linear sensing model given in~\eqref{eq:model}. The forward operator, denoted by $\Phib_{M_a}$, corresponds to a subsampling 2D Fourier measurement operator. The subsampling operator is a binary mask whose non-zero entries correspond to those pixels lying on a given number of lines, each of which contain the origin and are equally rotationally-spaced. We consider subsampling according to numbers of acquisition angles $M_a \in \{150, 200, 250, 300, 350\}$, corresponding to ratios $M/N \in \{0.52, 0.65, 0.76, 0.84, 0.90\}$. Examples are provided in Figure~\ref{fig:masks_map} for $M_a = 150$ and $M_a=250$. 
We generate the noise $\wb$ for the model \eqref{eq:model} as realizations of iid complex Gaussian noise with standard deviation 
$\delta = M^{-1}\|\Phib_{M_a} \overline{\xb}\|_2 10^{-\mathrm{iSNR} / 20}>0$, where  $\overline{\xb}$ is the ground truth image,
for input signal-to-noise ratio (iSNR) values in $ \{20, 25, 30, 35\}$ dB. 
Example MAP estimates are given in Figure~\ref{fig:masks_map}.



\begin{figure}[ht]
    \centering
    \includegraphics[width = 0.35\textwidth]{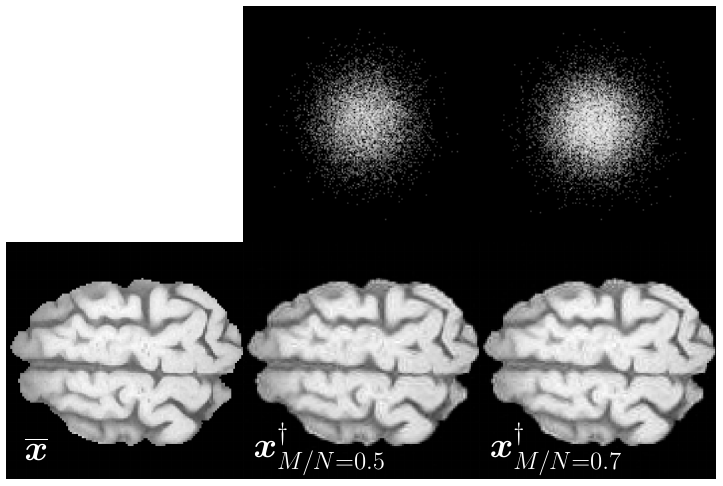}
    \caption{\small
    \textbf{Non-uniform Fourier with random sampling.} (top) Examples showing in Fourier domain the frequencies selected in obtaining the measurements, with $M/N = 0.5$ and $M/N = 0.7$, respectively. (bottom) Ground truth ($N = 128 \times 128$), and the two respective MAP estimates.}
    \label{fig:nufft_mri_masks_map}
\end{figure}

\paragraph{Non-uniform Fourier with random sampling}
For the second set of simulations, we continue using images from the same BRATS21 dataset. Here we consider a non-uniform Fourier sampling operator in the forward model, implemented using the DeepInverse library \cite{tachella2023deepinverse}. The sampling points in the Fourier domain are generated randomly from a Gaussian distribution with zero mean and variance equal to \(0.25\) of the maximum frequency. We consider sampling ratio values \(M/N \in \{0.3, 0.5, 0.7\}\). The noise \(\wb\) is generated as realizations of complex i.i.d Gaussian variables, with zero mean and variance \(\delta\) chosen to achieve iSNR values of \(\{20, 30, 40\}\) dB. Figure~\ref{fig:nufft_mri_masks_map} shows examples in the Fourier space where the frequencies are selected in obtaining the measurement, using random sampling with \(M/N = 0.5\) and \(M/N = 0.7\), alongside the two respective MAP estimates. 

\paragraph{Radon transform with limited angles}

Finally, we consider computed tomography (CT) imaging in our simulations where the measurement operator \(\Phib\) consists of a discretized Radon transform. The forward operator corresponds to a parallel-beam tomographic projection, with measurements acquired at distinct projection angles. To generate our CT simulation data, we use both the standard Shepp–Logan phantom and imaging data from the TCGA-LUAD collection \cite{albertina2016cancer}.

We select projection angle configurations consisting of \(\{30, 90, 120\}\) views, uniformly spaced over the full $180^\circ$, such that each finer sampling includes the previous coarser ones. We generate noise \(\wb\) as realizations of i.i.d Gaussian random variables, characterized by a standard deviation \(\delta\). The noise variance \(\delta\) is chosen to achieve iSNR values of \(\{30, 35, 40\}\) dB. Figure~\ref{fig:ct_sinogram_maps} shows examples of sinograms for projection angle of 90 and 120 views. 

\begin{figure}[t]
    \centering
    \includegraphics[width = 0.35\textwidth]{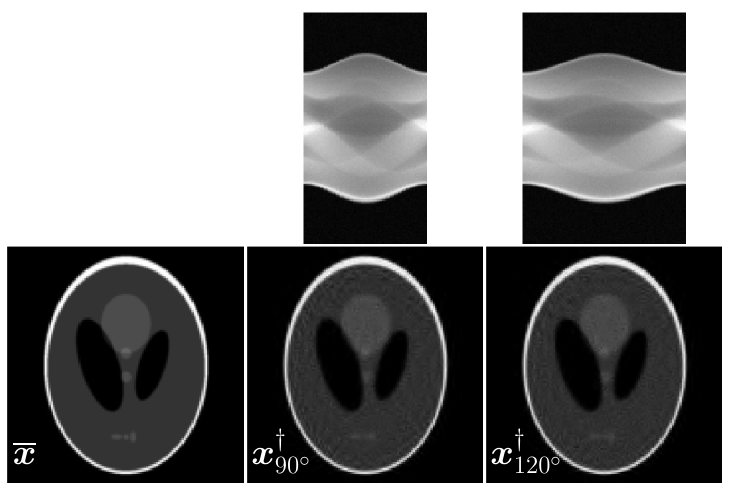}
    \caption{\small
    \textbf{Radon transform with limited angles (phantom)} (top) Examples showing sinograms with projection angle of 90 and 120 views, respectively. (bottom) Ground truth ($N = 120 \times 120$), and the two respective MAP estimates.}
    \label{fig:ct_sinogram_maps}
\end{figure}

Thereafter, structures of interest are identified manually and expressed as a binary mask $\Mb \in \{0,1\}^N$. 
The MAP estimate $\xb^\dagger$ and the mask $\Mb$ are then passed as inputs to Algorithm~\ref{alg:PD_Fix} to quantify uncertainty on structures of interest, as shown in Figure~\ref{fig:myGatedDnCNN}.

\begin{figure*}[t]
    \centering
    \begin{subfigure}[b]{0.7\textwidth}
    \includegraphics[width=1.05\linewidth]{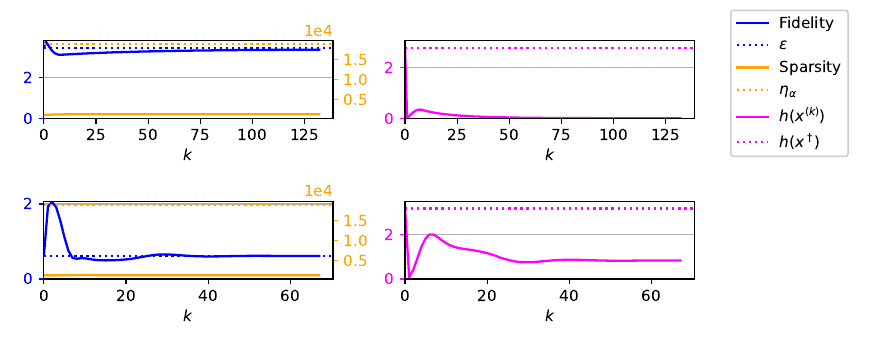}
    \vspace{-8mm}
    \caption{}\label{fig:data-lips-a}
    \end{subfigure}\hspace{-1.5mm}
    \begin{subfigure}[b]{0.7\textwidth}
    \includegraphics[width=1.05\linewidth]{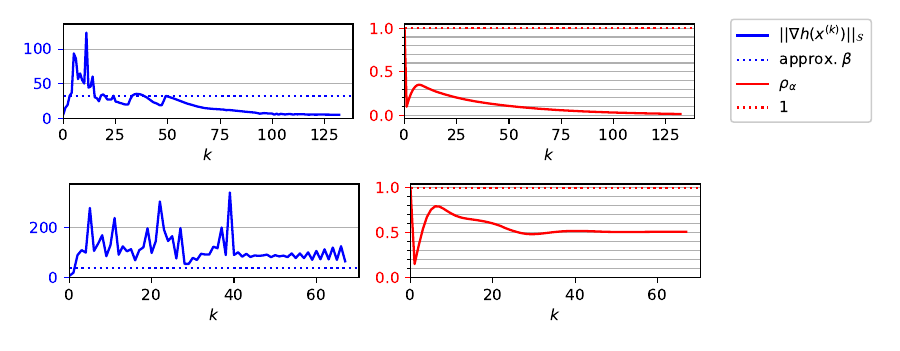}
    \vspace{-8mm}
    \caption{}\label{fig:data-lips-b}
    \end{subfigure}
    \vspace{-2mm}
    \caption{\small
    \textbf{Discrete Fourier sampling.} (a) Plots illustrate the data fidelity $f_\yb(\Phib \xb^{(k)})$, regularization $g(\xb^{(k)})$ and objective $h(\xb^{(k)})$. (b)  Plots illustrate $\rho_\alpha(\xb^{(k)}) = \sqrt{h(\xb^{(k)})/h(\xb^\dagger)}$ and the norm of the gradient $\|\nabla h(\xb^{(k)})\|_S$. Top (resp. bottom) plots correspond to $(M_\alpha, \mathrm{iSNR}) = (350, 20)$ (resp. $(350,35)$).}
    \label{fig:data-lips}
\end{figure*}

\section{Simulation results} \label{section:simulation-results}
\subsection{Algorithm behaviour analysis}
\label{Ssec:res:behav}

We first study the behaviour of the proposed PnP-BUQO, verifying its robustness across iterations, and comparing outputs with the original BUQO as well as its complexity. For the sake of simplicity, in this section we focus on the simple  case of discrete Fourier sampling described in Section~\ref{subsect:sim}.

\paragraph{Algorithm Robustness}

Figure~\ref{fig:data-lips} shows different curves to illustrate the convergence behavior of the proposed PnP-BUQO. 
Top (resp. bottom) plots correspond to the case $(M_a, \mathrm{iSNR})=(350, 20)$ (resp. $(M_a, \mathrm{iSNR})=(350, 35)$).
In Figure~\ref{fig:data-lips-a}, we show the convergence behavior of PnP-BUQO for constraints associated with $\tca$ (left), and for function $h$ (right). 
In Figure~\ref{fig:data-lips-b}, we show the evolution of $\| \nabla h(\xb^{(k)}) \|_S$ with iterations $k$ compared with the approximated estimation of $\beta$ used to choose the stepsizes in Algorithm~\ref{alg:PD_Fix} (left), and the evolution of the quantity $\rho_\alpha(\xb^{(k)}) = \sqrt{h(\xb^{(k)} / h(\xb^\dagger)}$ (right).


\begin{figure}[t]
\hspace{0.1in}\raisebox{0.15in}{\includegraphics[width = 0.13\textwidth, trim = 0.35in 0 0.35in 0, clip]{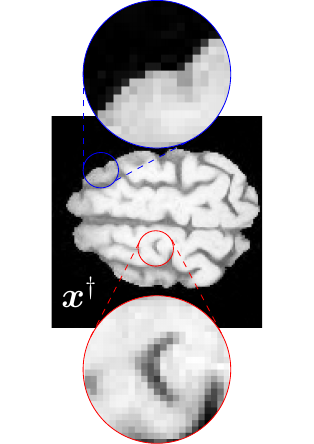}
}\includegraphics[width = 0.35\textwidth]{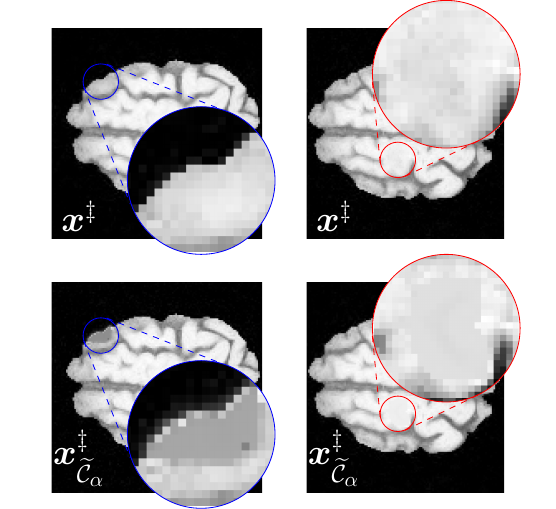}
    \caption{\small
    \textbf{Discrete Fourier sampling.} Comparison of BUQO with PnP-BUQO. 
    (left) MAP 
 $\xb^\dagger$ ($M_a = 200, \mathrm{iSNR} = 20$~dB). 
    (middle \& right) Results for two different structures, obtained with (top) PnP-BUQO $\xb^\ddagger$, ($\rho_\alpha = 0.018, \rho_\alpha = 0.018$) and (bottom) BUQO $\xb_\tca^\ddagger$ ($\rho_\alpha = 0.018, \rho_\alpha = 0.017$). In all instances, $\rho_\alpha < \tau$, so the data are inconclusive to reject $H_0$.}
    \label{fig:vanilla}
\end{figure}

\begin{figure*}[t]
\centering
    \includegraphics[width =0.95\textwidth]{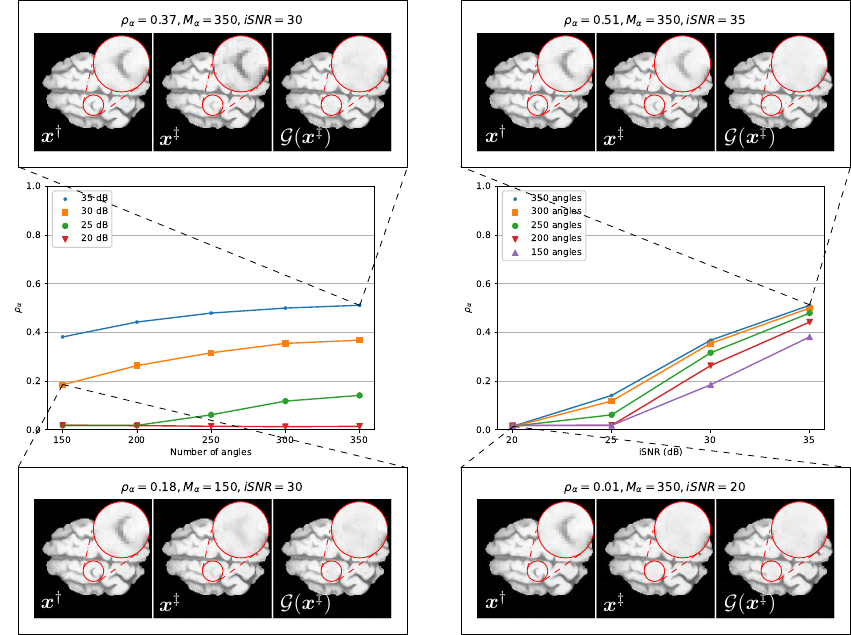}
    \caption{\small
    \textbf{Discrete Fourier sampling.} Structure confidence $\rho_\alpha$ (for a particular structure) with respect to the number of measurement angles, $M_a$ (left) and the iSNR values, i.e. noise level (right).}
    \label{fig:ra}
\end{figure*}

\paragraph{Comparison with classical BUQO}
We compare the qualitative fidelity of the structure-free image $\Gc(\xb^\ddagger)$ obtained with the proposed PnP-BUQO, with the results obtained with original BUQO \cite{RPW19}, for two structures, shown in Figure~\ref{fig:vanilla}. 
For both examples, we can observe that the proposed PnP-BUQO gives a more natural image than original BUQO, in particular when for edged structure regions. 
In the middle example we see that the original BUQO provides an image in $\tca$ that looks unnatural (darker patch in the inpainted region). Although this is less the case for the right example, BUQO still gives a patchy texture in the inpainted region. The result obtained by the original BUQO method appears cartoonish flat-textured due to the bounded energy term which defines the space of structure-free images. In contrast, PnP-BUQO produces images with a ``natural'' appearance.


\paragraph{Algorithm Complexity} 
A quantitative comparison of the time complexity is presented in Table~\ref{table:times}, for a particular structure of size $N_M = 170$ pixels. In our experiments, due to differentiating network $\Gc$ at each iteration, the proposed PnP-BUQO was slower than the original BUQO algorithm, when run on a unique structure. For this example, in total, PnP-BUQO was 4 times slower than BUQO. Note that similar observations were made in~\cite{Hurault2021} for a PnP algorithm where the involved neural network also requires back-propagation at each iteration. 
\begin{table}[b]
    \caption{\small
    \textbf{Discrete Fourier sampling.} Average computation time ($\pm$ std) for a structure containing $170$ pixels.}
    \label{table:times}
    \centering
    \begin{tabular}{l c c}
         & BUQO & PnP-BUQO \\
         \hline 
         \hline
         $\beta$ approximation per sample (s)& N/A & $ 4.3 \pm 1.9 $\\
         \hline
         Linear operator (s)& $3.6 \pm 0.2$& N/A\\ 
         \hline 
          Per iteration (ms)& $5.6\pm0.7$ & $19.7 \pm 2.2$ \\
         \hline
          $\nabla h$ computation (ms)& N/A & $15.2 \pm 1.8$ \\ 
         \hline
         Algorithm runtime (s) & $1.3 \pm 1.6$ & $3.8 \pm 3.5$ \\
         \hline
        Final iteration &$233\pm 297$ & $192 \pm 176$ \\
         \hline
         Total runtime (s)& $4.9\pm  1.7$ & $21.2 \pm 8.4$\\
         \hline
    \end{tabular}
\end{table}

However, as already emphasized in previous sections, the original BUQO requires an \textit{ad hoc} inpainting operator to be built per mask. In addition, the linear operator associated with the mask is built using a greedy approach, and hence its computation does not scale well with the size of the structure. 
In contrast, the proposed PnP-BUQO uses the same network for any mask, that does not need to be retrained. 

Note that choosing stepsizes for Algorithm~\ref{alg:PD_Fix} requires an approximation of $\beta$, as discussed in Section~\ref{subsect:params}, that incurs significant computational cost. However, the proposed approximate computation of $\beta$ only depends on the mask $\Mb$, and hence does not vary much for masks of similar size. So the approximation of $\beta$ can be computed in advance and used to perform multiple tests.

Overall, despite a significant increase in computation time per iteration for the proposed PnP-BUQO, removing the need to compute an \textit{ad hoc} linear inpainting operator and ignoring the cost of approximating $\beta$ would result in similar overall runtimes for BUQO and PnP-BUQO on average.

\subsection{Application to image restoration}\label{section:application-results}

The proposed PnP-BUQO not only improves the original BUQO, but also shows good generalization capability. In this section, we present applications of the proposed PnP-BUQO, where three different sets of simulation measurements data as described in Section~\ref{subsect:sim} are being considered. 
We show that even though the inpainting CNN is trained on the BRATS21 MRI dataset, its inpainting performance on CT images is also notable. 

\subsubsection{Discrete Fourier sampling}

Figure~\ref{fig:ra} illustrates the relationship between $\mathrm{iSNR}$ (i.e. noise level), the number of acquisition angles $M_\alpha$ and statistic $\rho_\alpha$ for the proposed PnP-BUQO. We observe that confidence in the presence of a structure increases monotonically with respect to both the number of acquisition angles and $\mathrm{iSNR}$.

We also consider both real structures that indeed exist in the ground truth image (Figure~\ref{fig:edge}), as well as restoration artifacts in the MAP estimates (Figure~\ref{fig:artefact}).
In Figure~\ref{fig:edge} we show the MAP estimate $\xb^\dagger$ (corresponding to $M_a = 150$ and iSNR=$30$ dB) and its structure-free version $\Gc(\xb^\dagger)$, as well as the PnP-BUQO output $\xb^\ddagger$ and its structure-free version $\Gc(\xb^\ddagger)$. In this case, the edge structure does exist in the ground truth image, and PnP-BUQO rejects $H_0$ hence confirming the existence of the structure.
In Figure~\ref{fig:artefact} the MAP estimate contains a checkerboard feature, which is not present in the ground truth. For this feature, simulations with $M_\alpha \in \{150, 200, 250, 300, 350\}$ all resulted in $\rho_\alpha \leq \tau$, hence the null hypothesis cannot be rejected, ensuring that the structure is certainly an artifact.
Finally, in Figure~\ref{fig:vanilla} we show two examples on running PnP-BUQO on both a true structure and a reconstruction artefact, where the data are inconclusive to reject $H_0$ (for $M_a=200$ and $\mathrm{iSNR}=20$~dB).


\begin{figure}[t]
    \centering
    \includegraphics[width = 0.43\textwidth]{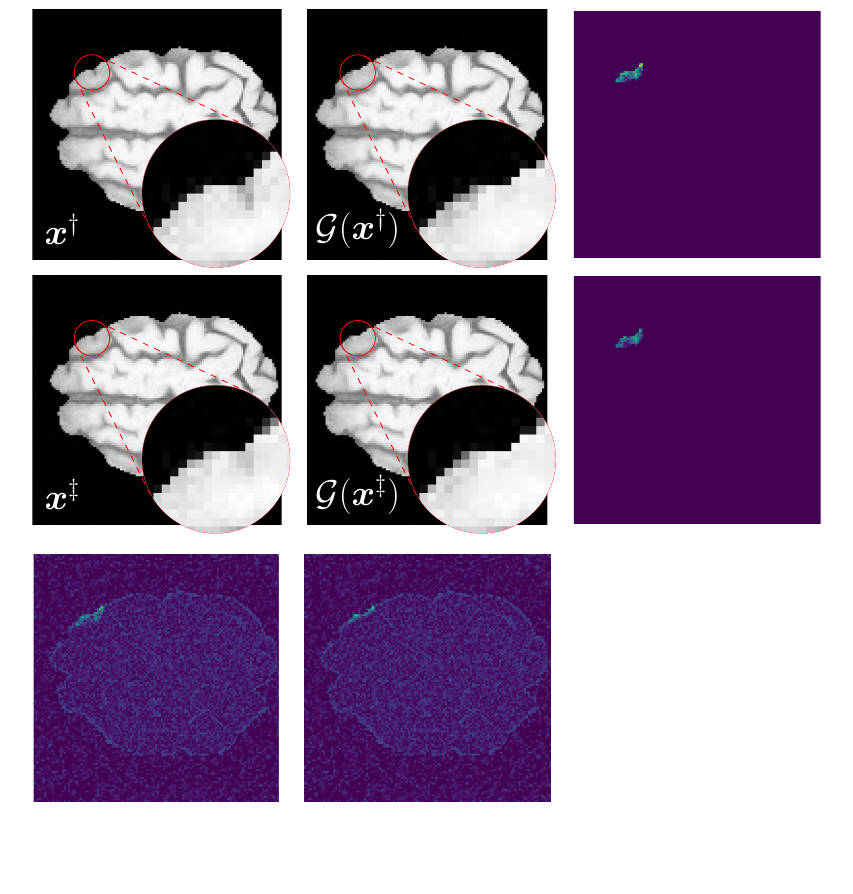}

    \vspace{-0.6cm}
    \caption{ \small
    \textbf{Dicrete Fourier sampling.} Example where the structure of interest is located near an edge (for $M_a = 150$, $\mathrm{iSNR} = 30$ dB, with $\rho_\alpha = 0.46$). 
    Grayscale images show (left to right, top to bottom) the MAP $\xb^\dagger$, the inpainted MAP $\Gc(\xb^\dagger)$, the output of the proposed PnP-BUQO $\xb^\ddagger$ and its inpainted version $\Gc(\xb^\ddagger)$.
    Color images show differences of images on horizontal and vertical axis, in log scale. 
    }
    \label{fig:edge}
\end{figure}

\begin{figure}[t]
  \centering
  \begin{subfigure}[b]{0.44\textwidth}
    \hspace*{0.85in}
    \includegraphics[%
      trim=0   0  1.325in 0,%
      clip,%
      width=0.33\textwidth%
    ]{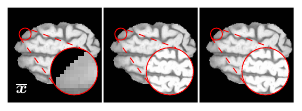}
    \vspace*{-0.075in}
  \end{subfigure}
  \hfill
    \begin{subfigure}[b]{0.44\textwidth}
    \includegraphics[trim=0 0 0 0.1in,%
      clip,%
      width=1\textwidth
    ]{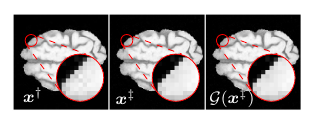}
  \end{subfigure}
  \caption{ \small
    \textbf{Discrete Fourier sampling.} (left) MAP (for $M_a=150$, iSNR=$25$ dB) with a checkerboard structure of interest. (center top) Ground truth. (center bottom) PnP-BUQO output $\xb^\ddagger$ and (right) inpainted output, with $\rho_\alpha = 0.018$.}
    \label{fig:artefact}
\end{figure}


\subsubsection{Non-uniform Fourier with random sampling}

We now focus on MRI measurements acquired through non-uniform Fourier with random sampling. 
In Figure~\ref{fig:nufft_mri_heatmap}, we give a heatmap summarizing the resulting $\rho_{\alpha}$ values obtained when questioning the true structure appearing in Figure~\ref{fig:nufft_mri_structure}, when varying measurement angles and noise conditions. As demonstrated by the heatmap, when the amount of uncertainty in the measurement data is high, either due to lower sampling ratios or high noise levels, the data support is resulting low.
To provide better quantitative measure of comparison of the two methods, in Figure~\ref{fig:nufft_mri_heatmap_ori_buqo}, we give a heatmap summarizing the $\widetilde{\rho}_{\alpha}$ values obtained with the original BUQO method, using same parameters as in Figure~\ref{fig:nufft_mri_heatmap}. 
The two heatmaps provided in Figures~\ref{fig:nufft_mri_heatmap_ori_buqo} and~\ref{fig:nufft_mri_heatmap} give very similar confidence levels for both methods. 
For the sake of completeness we also provide the associated PnP-BUQO and BUQO results in Figures~\ref{fig:nufft_mri_structure} and~\ref{fig:nufft_mri_structure_ori_buqo}, respectively, for the case $M/N=0.3$ and iSNR$=40$~dB. In this case $\rho_\alpha = 0.43$, hence rejecting $H_0$.

\begin{figure}[t]
\centering
    \includegraphics[width =0.35\textwidth]{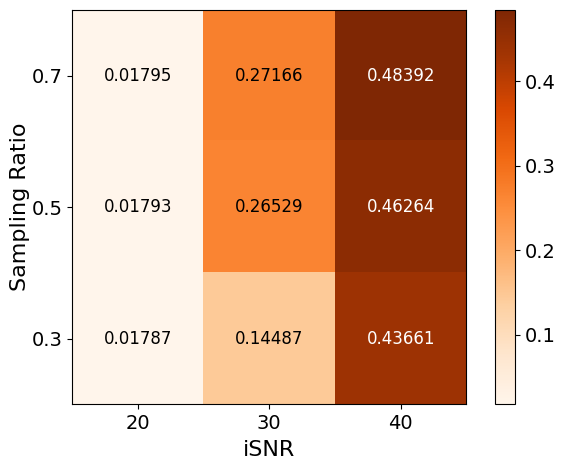}
    \caption{ \small 
    \textbf{Non-uniform Fourier with random sampling, with PnP-BUQO.} Heatmap showing the structure confidence $\rho_\alpha$ obtained with the proposed PnP-BUQO method (for the structure in Figure~\ref{fig:nufft_mri_structure}) with respect to sampling ratio (vertical) and iSNR values (horizontal).
    }
    \label{fig:nufft_mri_heatmap}
\end{figure}

\begin{figure}[t]
\centering
    \includegraphics[width =0.35\textwidth]{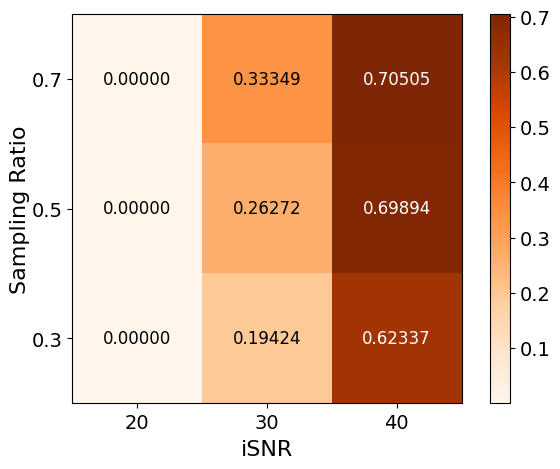}
    \caption{ \small 
    \textbf{Non-uniform Fourier with random sampling, with BUQO.} Heatmap showing the structure confidence $\widetilde{\rho}_\alpha$ obtained with the original BUQO method (for the structure in Figure~\ref{fig:nufft_mri_structure}) with respect to sampling ratio (vertical) and iSNR values (horizontal).
    }
    \label{fig:nufft_mri_heatmap_ori_buqo}
\end{figure}

\begin{figure}[t]
    \centering
    \includegraphics[width = 0.43\textwidth]{ 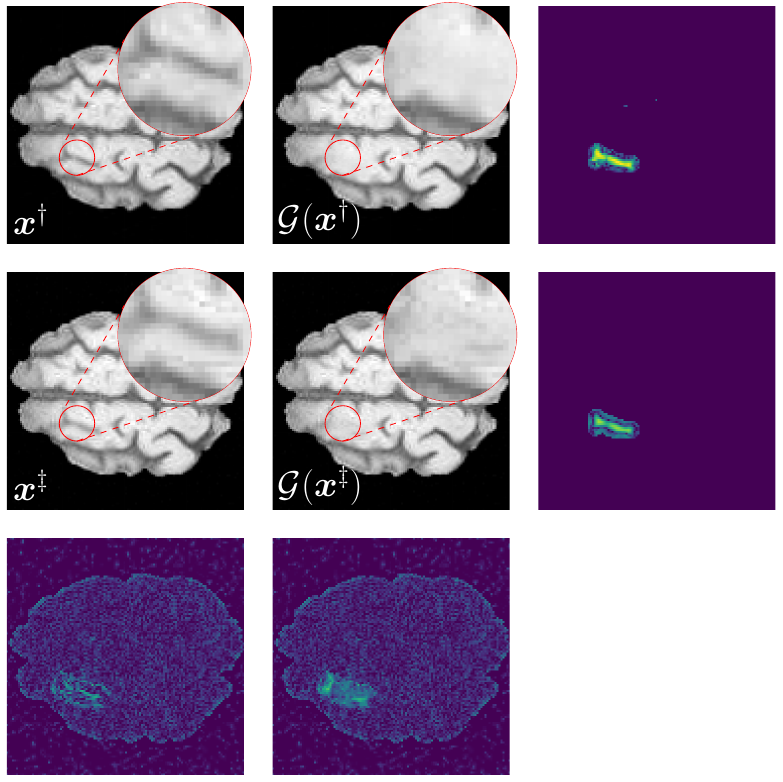}
    \caption{\small 
    \textbf{Non-uniform Fourier with random sampling, with PnP-BUQO.} Example where the structure of interest is located near the center left (for sampling ratio $M/N= 0.3$, $\mathrm{iSNR} = 40$ dB, with $\rho_\alpha = 0.43$). 
    Grayscale images show (left to right, top to bottom) the MAP $\xb^\dagger$, the inpainted MAP $\Gc(\xb^\dagger)$, the output of the proposed PnP-BUQO $\xb^\ddagger$ and its inpainted version $\Gc(\xb^\ddagger)$.
    Color images show differences of images on horizontal and vertical axis, in log scale. 
    }
    \label{fig:nufft_mri_structure}
\end{figure}

\begin{figure}[t]
    \centering
    \includegraphics[width = 0.43\textwidth]{ 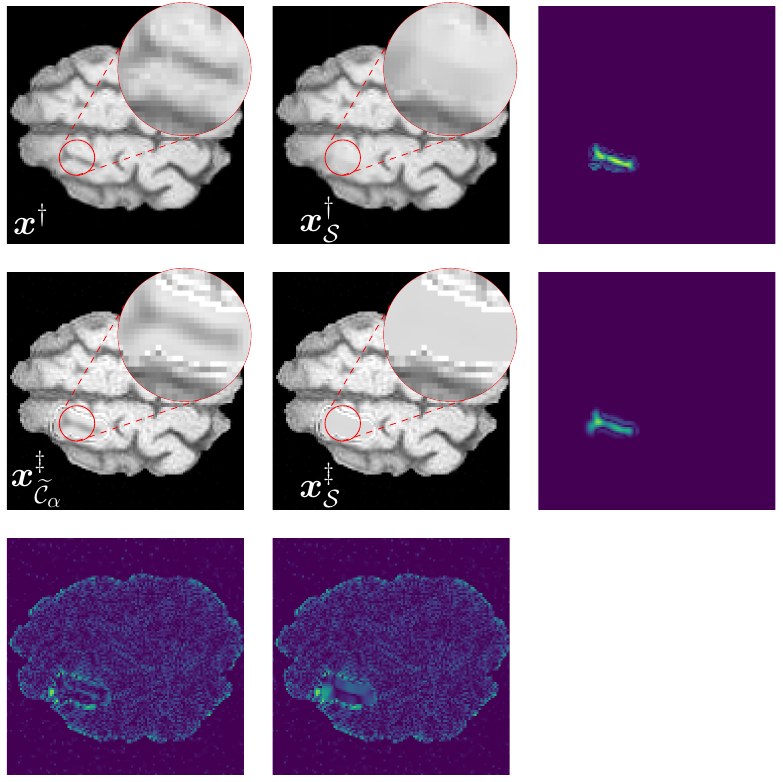}
    \caption{\small 
    \textbf{Non-uniform Fourier with random sampling, with BUQO.} Example obtained with the original BUQO method where the structure of interest is located near the center left (for sampling ratio $M/N= 0.3$, $\mathrm{iSNR} = 40$ dB, with $\rho_\alpha = 0.62$). Grayscale images show (left to right, top to bottom) the MAP $\xb^\dagger$, the structure-free MAP $\xb^\dagger_{\Sc}$, the output of the original BUQO method $\xb^\ddagger_{\tca}$ and $\xb^\ddagger_{\Sc}$.
    Color images show differences of images on horizontal and vertical axis, in log scale. 
    } 
    \label{fig:nufft_mri_structure_ori_buqo}
\end{figure}



\subsubsection{Radon transform with limited angles}

To illustrate the generalizability of the proposed method, we also validate the performance of PnP-BUQO on CT imaging data, considering the model described in Section~\ref{subsect:sim}.

We record the $\rho_\alpha$ values under various measurement conditions, summarized in a heatmap given in Figure~\ref{fig:ct_phantom_heatmap}, for the true structure shown in Figure~\ref{fig:ct_phantom_structure}. Results are obtained varying tomography measurement angles \(\alpha \in \{30, 90, 120\}\) and noise levels (iSNR \(\in \{20, 25, 30\}\) dB). Results again illustrate a monotonic relationship between the measurement uncertainty and the amount of support by the data. For visual inspection we show associated PnP-BUQO results for Figure~\ref{fig:ct_phantom_structure} when considering $90$ angles and iSNR~$=35$~dB. In this case $\rho_\alpha=0.536$, hence $H_0$ can be rejected.

Finally, for the sake of completeness, we test the PnP-BUQO method on a realistic CT image taken from the TCGA-LUAD data collection \cite{albertina2016cancer}
, shown in Figure~\ref{fig:ct_scan_structure}.  
Here as well we test for a true structure, and the resulting $\rho_\alpha=0.27$ suggests $H_0$ can be rejected.

\begin{figure}[t]
\centering
    \includegraphics[width =0.35\textwidth]{ 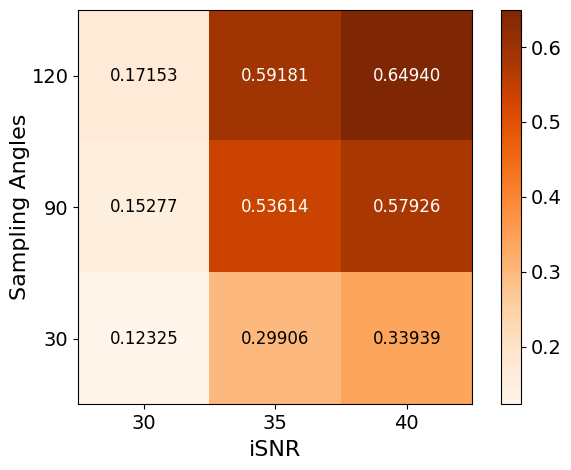}
    \caption{\small
    \textbf{Radon transform with limited angles (phantom)} Heatmap showing the structure confidence $\rho_\alpha$ (for the structure in Figure~\ref{fig:ct_phantom_structure}) with respect to sampling ratio (vertical) and measurement noise (horizontal).
    }
    \label{fig:ct_phantom_heatmap}
\end{figure}

\begin{figure}[ht]
    \centering
    \includegraphics[width = 0.43\textwidth]{ 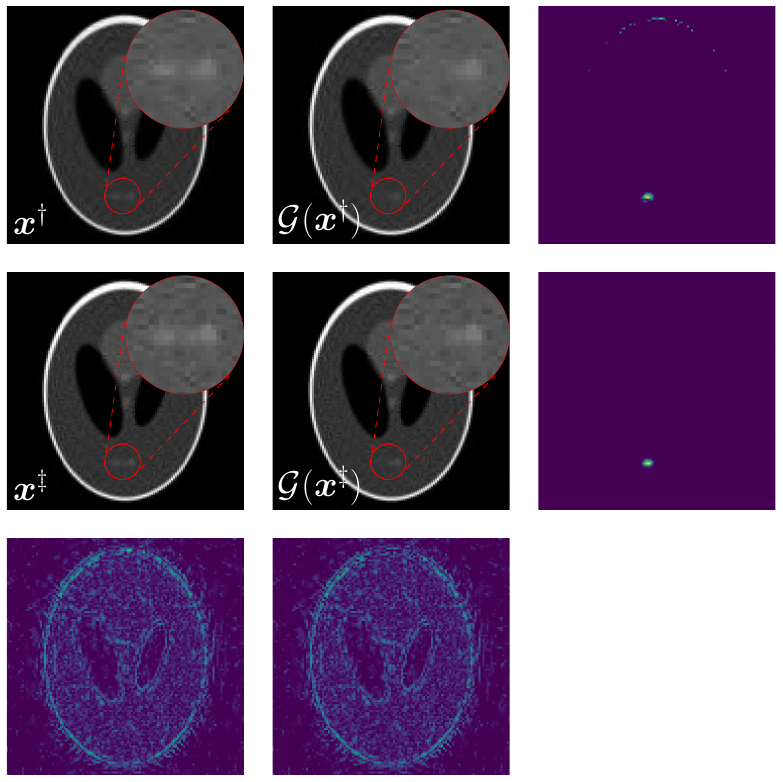}
    \caption{ \small
    \textbf{Radon transform with limited angles (phantom)} Example where the structure of interest is located near an edge (for 90 sampling angles, $\mathrm{iSNR} = 35$ dB, with $\rho_\alpha = 0.536$). 
    Grayscale images show (left to right, top to bottom) the MAP $\xb^\dagger$, the inpainted MAP $\Gc(\xb^\dagger)$, the output of the proposed PnP-BUQO $\xb^\ddagger$ and its inpainted version $\Gc(\xb^\ddagger)$.
    Color images show differences of images on horizontal and vertical axis, in log scale. 
    }
    \label{fig:ct_phantom_structure}
\end{figure}

\begin{figure}[ht]
    \centering
    \includegraphics[width = 0.43\textwidth]{ 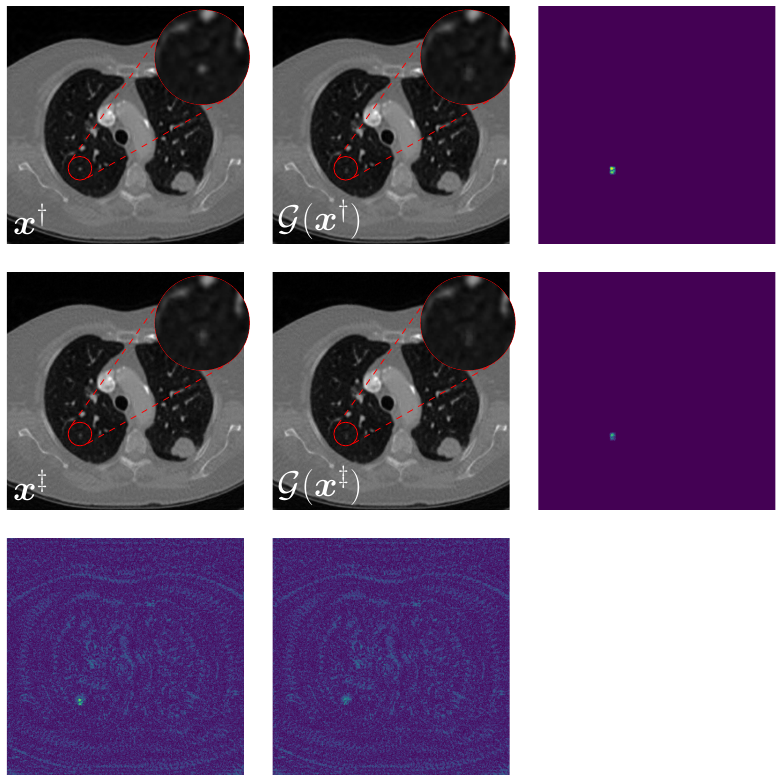}
    \caption{\small
    \textbf{Radon transform with limited angles (scan)} Example where the structure of interest is located near the center bottom (for 90 sampling angles, $\mathrm{iSNR} = 40$ dB, with $\rho_\alpha = 0.27$). 
    Grayscale images show (left to right, top to bottom) the MAP $\xb^\dagger$, the inpainted MAP $\Gc(\xb^\dagger)$, the output of the proposed PnP-BUQO $\xb^\ddagger$ and its inpainted version $\Gc(\xb^\ddagger)$.
    Color images show differences of images on horizontal and vertical axis, in log scale. 
    }
    \label{fig:ct_scan_structure}
\end{figure}

\section{Conclusion}\label{section:discussion}


In this work, we introduce PnP-BUQO, a plug-and-play Bayesian uncertainty quantification method built on primal-dual Condat-V\~{u} iterations, and empowered by a learned inpainting operator $\Gc$. 
The inpainting operator defines structure-free images through a fixed-point formulation. Such a formulation paves the way to investigate arbitrary structures and imaging modalities without bespoke operator design.
In particular, allowing $\Gc$ to be an arbitrary operator sets our framework free to apply in great generality. Moreover, it allows for the user to create custom operators to adapt the definition of $\Sc_\tau$ to any particular application. 

We carried out numerical experiments where the proposed approach has been tested on two Fourier undersampling problems (discrete and non-uniform) inspired by magnetic resonance imaging (MRI), as well as a computed tomography (CT) problem using the Radon measurement operator. Our analysis shows that the proposed PnP-BUQO provides more natural visual quality images than the original BUQO method \cite{RPW18, RPW19}. Notably, despite that the inpainting network is trained solely on brain MR images, our framework transfers seamlessly to CT imaging data, confirming its cross-modal generality.


One future aspect of the work would be on further improving its computational efficiency. As an example, the time complexity of Algorithm~\ref{alg:PD_Fix} may be reduced by removing backpropogation as in~\cite{Hurault2021}, or using a more lightweight NN, which preserves robustness and fidelity.

\section*{Acknowledgments}
The work of M. Tang was partly supported by the EPSRC grant EP/W522648/1 and A. Repetti was partly supported by the Royal Society of Edinburgh. The work of all the authors was partly supported by the EPSRC grant EP/X028860.


\appendix

\section{Hand-crafted structure definition} \label{appendix:struct}

As explained in Section~\ref{Ssec:BUQO-var-form}, the set $\Sc$ was defined in \cite{RPW18, PNAS2022} as an intersection of three convex sets, i.e., $\Sc = \Sc_1 \inter \Sc_2 \inter \Sc_3$, where, for every $s \in \{1, \ldots, 3\}$, $\Sc_s$ is convex, closed and proper.
In both articles, $\Sc_1 = [0,+\infty)^N$ to enforce positivity of intensity images. 
The second set was used to \textit{remove} the structure by controlling the energy in the structure area, i.e., $\Sc_2 = \menge{\xb \in \RR^N}{ \Mb \xb \in \Bc_2(\mu, \theta)}$, for some hyper-parameters $( \mu, \theta) \in [0,+\infty)^2$ that need to be manually fine-tuned.
Finally the third set was used to \textit{inpaint} the structure area, using information from its neighborhood. In particular, in \cite{RPW18} the authors used convolutions, while in \cite{PNAS2022} a total variation (TV) type smoothing was used. In both the cases, $\Sc_3$ also depends from hyper-parameters.

For the sake of completeness, we explain below how $\Sc_3$ was built in \cite{RPW18}, as this is used in our simulations. Specifically, the authors define $\Sc_3 = \menge{\xb \in \RR^N}{ \Mb \xb - \Lb \Mb^c \xb \in [-\beta, +\beta] }$, where $\beta>0$ 
and $\Lb \colon \RR^{N - N_\Mb} \to  \RR^{N_\Mb} $ a linear inpainting operator that consists in computing convolutions starting from the edges of the structure and iterating in the direction of its centre.
Building such an operator relies on a greedy technique, which is not scalable. It is further structure dependent, and must be rebuilt when assessing another structure, hence leading again to scalability issues. The TV-ball approach proposed in \cite{PNAS2022} does not have this issue, but can only be applied to piece-wise constant images, possibly causing ``patches'' to appear in $\xb_\Sc^\ddagger$.

\section{Algorithm for MAP estimate} \label{appendix:map}

As defined in Section \ref{subsec:map_estimate}, the resulting particular instance of problem~\eqref{eq:MAP} can be solved with primal-dual algorithms (see, e.g., \cite{Condat13, Vu13, Komodakis2014, combettes2012}). We propose to solve it with the primal-dual Condat-V\~u method, described in Algorithm~\ref{alg:PD_MAP} (where $ \| \cdot \|_S$ denotes the spectral norm). The output of Algorithm~\eqref{alg:PD_MAP} is denoted by $\xb^\dagger$.

\begin{algorithm}[!h]
\caption{Primal-Dual algorithm for MAP estimate}\label{alg:PD_MAP}
\begin{algorithmic}
\STATE{\textbf{Initialization:}  
Let $\xb^{0} \in \RR^N$, $\vb_1^{(0)} \in \RR^P$ and $\vb_2^{(0)} \in \eC^M$. \\
Let $(\mu_1, \mu_2, \sigma) \in [0,+\infty[^3$ be such that $ \sigma \big( \mu_1 \| \Psi \|_S^2 + \mu_2 \| \Phi \|_S^2 \big) < 1 $.
}
\STATE{\textbf{Iterations:}}
 \FOR{$k=0, 1, \ldots$}
  \STATE{$\widetilde{\vb}_1^{(k)} = \vb_1^{(k)} + \mu_1 \Psib \xb^{(k)}$}
  \STATE{${\vb}_1^{(k+1)} = \widetilde{\vb}_1^{(k)} - \mu_1 \mathrm{prox}_{\mu_1^{-1}\| \cdot \|_1}(\widetilde{\vb}_1^{(k)} / \mu_1) $}
  \STATE{$\widetilde{\vb}_2^{(k)}  = \vb_2^{(k)} + \mu_2 \Phib \xb^{(k)}$}
  \STATE{$\vb_2^{(k+1)} = \widetilde{\vb}_2^{(k)} - \mu_2 \Pi_{\Bc_2(\yb, \varepsilon)}(\mu_2^{-1}\widetilde{\vb}_2^{(k)})$}
  \STATE{$\widetilde\xb^{(k)} = \Pi_{[0,1]^N}\Big(\xb^{(k)}- \sigma \big( \Psib^\dagger \vb_1^{(k+1)} + \Phib^\dagger \vb_2^{(k+1)}\big)\Big)$}
  \STATE{$\xb^{(k+1)} = 2\widetilde\xb^{(k)} - \xb^{(k)}$}
 \ENDFOR
\end{algorithmic}
\end{algorithm}

\section{Primal-Dual algorithm for BUQO} \label{appendix:buqo}

In this section we detail the approach proposed in \cite{RPW18} to perform the hypothesis test corresponding to BUQO as described in Section~\ref{Ssec:BUQO-hyp}. With the variational formulation of set $\Sc$ given in Section~\ref{Ssec:BUQO-var-form}, we proceed to solve the minimization problem~\eqref{eq:min}. Following \eqref{eq:Calpha-balls}, both sets $\tca$ and $\Sc$ are defined as the intersection of multiple convex sets. Hence, a primal-dual splitting algorithm was proposed in~\cite{RPW19} to solve Problem~\eqref{eq:min}.  

We detail the iteration steps in Algorithm~\ref{alg:PD_BUQO}, using the same notation as in Section~\ref{Ssec:BUQO-var-form}. In particular, as emphasized earlier, the operator $\Lb$ is dependent on the choice of the mask $\Mb$, and must be adjusted (or even recomputed \cite{RPW18, RPW19}) for different choices of this mask. 

Let $(\xb_{\tca}^\ddagger, \xb_{\Sc}^\ddagger) \in \tca \times \Sc$ be the output of Algorithm~\ref{alg:PD_BUQO}. Then the hypothesis test is performed using Corollary~\ref{cor:1}, by evaluating $\| \xb_{\tca}^\ddagger - \xb_{\Sc}^\ddagger \|_2 $.
If the resulting quantity is greater than $0$, up to some accepted tolerance (see \cite{RPW18, RPW19}), then we conclude by Corollary~\ref{cor:1} that the null hypothesis $H_0$ may be rejected, with confidence $\alpha$.

\begin{algorithm}[h]
\caption{Primal-Dual algorithm for BUQO}\label{alg:PD_BUQO}
\begin{algorithmic}
\STATE{\textbf{Initialization:} 
Let $(\xb_{\tca}^{(0)}, \xb_{\Sc}^{(0)}) \in ([0,\infty[^{N})^2$, 
$(\vb_1^{(0)}, \vb_2^{(0)}) \in \eR^P \times \eC^M$, 
$(\ub_2^{(0)}, \ub_1^{(0)}) \in (\eR^{N_M})^2$ and $\ub_2^{(0)} \in\eR^{N_M}$. \\
Let $\mu_{1,1}, \mu_{1,2}, \mu_{2, 1}, \mu_{2, 2}, \sigma>0$ be such that $\sigma^{-1} - \mu_{1,1} \|\Psi\|_S^2 - \mu_{1,2} \|\Phi\|_S^2 - \mu_{2,1} \|\overline{\Lb}\|_S^2 - \mu_{2,2} > \gamma/2 $. 
}
\STATE{\textbf{Iterations:}}
 \FOR{$k=0, 1, \ldots$}
  \STATE{$\widetilde{\vb}_1^{(k)} = \vb_1^{(k)} + \mu_{1,1} \Psib \xb_{\tca}^{(k)}$}
  \STATE{$\vb_1^{(k+1)} = \widetilde {\vb}_1 - \mu_{1,1} \Pi_{\Bc_1(\zerob, \widetilde{\eta}_\alpha/\lambda)} \Big(\mu_{1,1}^{-1} \widetilde{\vb}_1^{(k)} \Big) $}
  \STATE{$\widetilde{\vb}_2^{(k)}= \vb_2^{(k)} + \mu_{1,2} \Phib \xb_{\tca}^{(k)}$}
  \STATE{$\vb_2^{(k+1)} = \widetilde{\vb}_2 - \mu_{1,2} \Pi_{\Bc_2(\yb, \varepsilon)} \Big( \mu_{1,2}^{-1} \widetilde{\vb}_2^{(k)} \Big) $}
  \STATE{$ \!\!\! \begin{array}{ll}
                \widetilde{\xb}_{\tca}^{(k)} \!\!\!\!  &= \Pi_{[0,1]^N}\Big( (1-\gamma \sigma) \xb_\tca^{(k)} + \gamma \sigma \xb_{\Sc}^{(k)} \\
                                    & \quad \quad \quad \quad \quad \quad - \sigma \Psib^\dagger \vb_1^{(k+1)}  - \sigma \Phib^\dagger \vb_2^{(k+1)} \Big)
            \end{array}$}
    \STATE{$ \xb_{\tca}^{(k+1)}  = 2 \widetilde{\xb}_{\tca}^{(k)} -  \xb_{\tca}^{(k)}$\\[0.2cm]}
  \STATE{$\widetilde {\ub}_1^{(k)} = \ub_1^{(k)} + \mu_{2,1} \overline{\Lb} \xb_{\Sc}^{(k)}$}
  \STATE{$\ub_1^{(k+1)} = \widetilde{\ub}_1^{(k)} - \mu_{2,1} \Pi_{[-\tau, \tau]^{N_M}} \Big( \mu_{2, 1}^{-1} \widetilde{\ub}_1^{(k)} \Big) $}
  \STATE{$\widetilde{\ub}_2^{(k)} = \ub_2^{(k)} + \mu_{2, 2} \Mb\xb_{\Sc}^{(k)}$}
  \STATE{$\ub_2^{(k+1)} = \widetilde{\ub}_2^{(k)} - \mu_{2, 2} \Pi_{\Bc_2 (\mu, \theta)}\Big( \mu_{2, 2}^{-1} \widetilde{\ub}_2^{(k)} \Big) $}
  \STATE{$ \!\!\! \begin{array}{ll}
                \widetilde{\xb}_{\Sc}^{(k)} \!\!\!\!  &= \Pi_{[0,1]^N} \Big( (1-\gamma \sigma) \xb_{\Sc}^{(k)} + \gamma \sigma \xb_\tca^{(k)} \\
                                    &   \quad \quad \quad \quad \quad \quad -\sigma \overline{\Lb}^\dagger \ub_1^{(k+1)}   - \sigma \Mb^\dagger \ub_2^{(k+1)}  \Big)
            \end{array}$}
    \STATE{$ \xb_{\Sc}^{(k+1)}  = 2 \widetilde{\xb}_{\Sc}^{(k)} -  \xb_{\Sc}^{(k)}$}
 \ENDFOR
\end{algorithmic}
\end{algorithm}

{\bibliography{sn-bibliography}}

\end{document}